\documentclass[journal]{IEEEtran}
\usepackage[pdftex,breaklinks,colorlinks,
citecolor=blue,
urlcolor=blue]{hyperref}

\usepackage{cite}
\usepackage{graphicx}

\usepackage{amsmath}
\usepackage{empheq}
\usepackage{bbm}
\usepackage{amssymb}
\usepackage{mathrsfs}  
\usepackage{amsthm}
\usepackage{bm}

\usepackage{array}

\usepackage{url}

\usepackage[braket, qm]{qcircuit}

\usepackage[textsize=small]{todonotes}

\usepackage{braket}

\newcommand{\CC}{\mathbb{C}}
\newcommand{\EE}{\mathbb{E}}
\newcommand{\NN}{\mathbb{N}}
\newcommand{\PP}{\mathbb{P}}
\newcommand{\PPb}{\overline{\mathbb{P}}}
\newcommand{\RR}{\mathbb{R}}
\newcommand{\bV}{\mathbf{V}}
\newcommand{\VV}{\mathbb{V}}
\newcommand{\bU}{\mathbf{U}}

\newcommand{\D}{\mathrm{d}}
\newcommand{\E}{\mathrm{e}}
\newcommand{\I}{\mathrm{i}}

\newcommand{\Cc}{\mathcal{C}}
\newcommand{\Ff}{\mathcal{F}}
\newcommand{\Ffov}{\overline{\mathcal{F}}}
\newcommand{\Hh}{\mathcal{H}}

\newcommand{\nq}{\mathfrak{n}}
\newcommand{\nqb}{\overline{\mathfrak{n}}}
\newcommand{\pf}{\mathfrak{p}}
\newcommand{\LOnefhat}{L^1[\widehat{f}]}

\newcommand{\half}{\frac{1}{2}}

\newcommand{\ww}{\boldsymbol{w}}
\newcommand{\xx}{\boldsymbol{x}}
\newcommand{\bx}{\boldsymbol{\xi}}
\newcommand{\ab}{\boldsymbol{a}}
\newcommand{\bb}{\boldsymbol{b}}
\newcommand{\gb}{\boldsymbol{\gamma}}
\newcommand{\Af}{\boldsymbol{\mathbf{A}}}
\newcommand{\Bf}{\boldsymbol{\mathbf{B}}}
\newcommand{\Wf}{\boldsymbol{W}}
\newcommand{\ttheta}{\boldsymbol{\theta}}
\newcommand{\TTheta}{\boldsymbol{\Theta}}
\newcommand{\ftrn}{f_{n,\ttheta}^{R}}
\newcommand{\gtrn}{g_{n,\ttheta}^{R}}
\newcommand{\Lbarf}{\overline{L}_{2}[f]}

\newcommand{\Ug}{\mathtt{U}}
\newcommand{\Ig}{\mathtt{I}}
\newcommand{\Cg}{\mathtt{C}}

\newcommand{\Hg}{\mathtt{H}}

\newcommand{\Rgy}{\mathtt{R}_{\mathrm{y}}}
\newcommand{\Rgz}{\mathtt{R}_{\mathrm{z}}}

\newcommand{\Vg}{\mathtt{V}}

\newcommand{\R}{\mathbb{R}}

\newcommand{\C}{\mathbb{C}}

\newcommand{\N}{\mathbb{N}}

\renewcommand{\P}{\mathbb{P}}

\newtheorem{theorem}{Theorem}[section]
\newtheorem{corollary}[theorem]{Corollary}
\newtheorem{lemma}[theorem]{Lemma}
\newtheorem{proposition}[theorem]{Proposition}

\theoremstyle{definition}
\newtheorem{definition}[theorem]{Definition}
\newtheorem{remark}[theorem]{Remark}
\newtheorem{example}[theorem]{Example}

\begin{document}

\title{Universal Approximation Theorem and Error Bounds for Quantum Neural Networks and Quantum Reservoirs}

\author{Lukas~Gonon and
	Antoine~Jacquier
	\thanks{L. Gonon is with the School of Computer Science, University of St. Gallen, Switzerland (e-mail: lukas.gonon@unisg.ch). A. Jacquier is with the Department of Mathematics, Imperial College London, UK (e-mail: a.jacquier@imperial.ac.uk). 
	A.  Jacquier  is also affiliated with the the Alan Turing Institute.}
}

\markboth{Gonon and Jacquier: Universal Approximation Theorem for Quantum Neural Networks}%
{Gonon and Jacquier: Universal Approximation Theorem for Quantum Neural Networks}
%


\maketitle


\begin{abstract}
Universal approximation theorems are the foundations of classical neural networks,
providing theoretical guarantees that the latter are able to approximate maps of interest.
Recent results have shown that this can also be achieved in a quantum setting,
whereby classical functions can be approximated by parameterised quantum circuits.
We provide here precise error bounds for specific classes of functions
and extend these results to the interesting new setup of randomised quantum circuits,
mimicking classical reservoir neural networks. Our results show in particular that a quantum neural network with $\mathcal{O}(\varepsilon^{-2})$ weights and  $\mathcal{O} (\lceil \log_2(\varepsilon^{-1}) \rceil)$ qubits suffices to achieve {approximation error} $\varepsilon>0$ when approximating  functions with integrable Fourier transforms.
\end{abstract}


%
\IEEEpeerreviewmaketitle

\section{Introduction}
%
%
%
%
\IEEEPARstart{A}{rtificial} neural networks were devised decades ago to approximate arbitrary functions by some algorithmically generated composition of maps.
With an increasing level of generality, {Kolmogorov~\cite{Kolmogorov1957}, Arnold~\cite{Arnold1957},} Cybenko~\cite{cybenko1989approximation}, Hornik~\cite{hornik1991}, Hornik, Stinchcombe, White~\cite{hornik1989multilayer}, Leshno, Lin, Pinkus, Schocken~\cite{leshno1993multilayer} proved that spaces of artificial neural networks (with arbitrary width) are in fact dense within some spaces of functions, 
giving rise to the notion of \emph{universal approximations}.
The case of arbitrary depth and fixed width was studied later by 
Gripenberg~\cite{gripenberg2003approximation} and Kidger and Lyons~\cite{kidger2020universal}{,} among many others. Bounds on the approximation errors have been obtained by Barron~\cite{Barron1992, Barron1993, Barron1994ApproximationAE}, Mhaskar~\cite{Mhaskar1996} and in many other works, for instance in~\cite{Boelcskei2019, Guehring2020, yarotsky2017error}. We refer to \cite{Guehring2023} for a review of existing results.
These results---and many subsequent papers refining them---provide theoretical grounds for the use of neural networks in applications.
While the generated function spaces are not constructive, they nevertheless guarantee that complicated and seemingly intractable functions (possibly in non-Euclidean or even infinite-dimensional spaces) can be well approximated by easy-to-implement neural networks, the simplest of them being feedforward neural networks.

Over the past few years, the rise of quantum computing capabilities, while still in their infancy, 
has opened the gates to (small-scale) applications, in particular optimisation 
(through quantum annealing, as developed by \texttt{D-Wave}) and machine learning.
There has been wide interest in searching for so-called \emph{quantum supremacy}, 
or at least (and more realistically) \emph{quantum advantage} 
(a term coined by Preskill in~\cite{preskill2012quantum}) and the many empirical results available at the moment, 
in biology~\cite{cordier2022biology} or in finance~\cite{stamatopoulos2022towards}, for example, 
indicate that more research is needed to provide clear benefits for real-life applications.
In the context of machine learning specifically, 
quantum computing essentially replaces layers of feedforward neural networks by parameterised quantum circuits; the similarities end here as the non-linearity of classical activation functions (such as {tanh or ReLU}) find no specific equivalent in quantum neural networks, 
and a different approach is thus required. 
The fundamental difference of outputs 
(a vector on the real line for a classical neural network and a discrete probability distribution for its quantum version) also requires a different thought process.

{Weight randomisation is a popular technique in classical machine learning, employed to circumvent potential challenges arising from the high-dimensional parameter space. This paradigm, at the core of classical random features \cite{RahimiRecht2008a, RahimiRecht2008}, extreme learning machines~\cite{HCS2006} and reservoir computing~\cite{Jaeger2001, LukoseviciusJaeger2009, TANAKA2019100}), gives the advantage of reducing the dimension of the optimisation problem, thereby being more amenable to applications. Universality results for such systems have been obtained in~\cite{RC8, RC20, GrigOrtega2018, GrigoryevaOrtega2018, Gonon2021, RC12}.
Analogously to the classical case, weight randomisation has also been employed in quantum machine learning.
Quantum reservoir computing~\cite{FujiiNakajima2017} and    quantum extreme learning machines have been extensively studied in the past years{, both} from a theoretical and from an empirical perspective. {We refer, for example, to} \cite{Dasgupta2020, MPOrtega2022, MOLTENI2023, Suzukietal2021} and the review papers~\cite{Ghoshetal2021, Mujaletal2021}.}

The present paper focuses on the theoretical aspects of quantum neural networks{, quantum random neural networks and quantum extreme learning machines}.
{A universal approximation theorem for quantum neural networks was obtained in} \cite{PerezSalinas2020datareuploading,perez2021one}
by constructing a one-qubit quantum circuit able to arbitrarily approximate any continuous complex-valued function.
A similar approach was carried out in~\cite{Schuld2021}, {showing} that data encoding can be approximated by infinitely repeating (akin to infinite-width classical neural networks)  simple encoding schemes based on Pauli gates.
Both papers rely on Fourier series representations of the function to be approximated,
a very natural path because of their trigonometric interpretations, similar to the actions of (quantum) rotation matrices.  
{In a dynamic setting, universal approximation properties have  been established for various types of quantum reservoir computing systems. In particular, it has been shown that fading-memory input-output systems can be approximated universally, using  certain types of dissipative quantum systems \cite{Chen2019,Chenetal2020}, and using a Gaussian quantum harmonic oscillator network \cite{Nokkala2021}.
}

{However, in contrast to the case of classical neural networks{,} for quantum neural networks no approximation error bounds have been available in the previous literature {neither in the static nor the dynamic setting}. In particular, there has been no understanding on the  number of qubits and the size of the quantum circuit that is required in order to guarantee a certain {approximation} accuracy. 
The goal of the present article is to move one step further and to prove precise error bounds for these approximations {in the static setting}. We cover both the case of trainable quantum neural networks and the case of quantum random neural networks or quantum extreme learning machines. 
}
{The question as to whether quantum random features overpower their classical counterparts is still an open question.
The genesis of this paper is anchored in the hope---supported by some empirical evidence~\cite{Schuld2021, wu2021expressivity, wu2024expressivity, wu2024randomness}---that quantum neural networks provide more expressivity, 
and we leverage this randomisation aspect from classical neural networks to facilitate the use of such quantum networks,
via dimension reduction.
}

{Our first main contribution}
consists of error bounds for a universal approximation theorem for continuous functions, bounded in~$L^{1}$ with mild constraints on their Fourier transforms. 
To do so, we explicitly build a parameterised quantum circuit and prove that a set of hyperparameters (or rotation angles) can achieve accurate estimation. More precisely, we show that a quantum neural network with $\mathcal{O}(\varepsilon^{-2})$ weights and  $\mathcal{O} (\lceil \log_2(\varepsilon^{-1}) \rceil)$ qubits suffices to achieve {approximation error at most}  $\varepsilon>0$ when approximating  functions with integrable Fourier transform.
{In particular, in order to refine the approximation accuracy the required additional number of qubits only grows very slowly. }
{Our second main contribution then consists in analogous results for quantum random neural networks. Thus, the original quantum circuit is replaced by a reservoir quantum circuit, where all the unitary operators apart from the last one, are randomised and frozen. We also show that these circuits are able to approximate a large class of continuous functions at {approximation error}  $\varepsilon>0${,} with only $\mathcal{O}(\varepsilon^{-2})$ random weights and with the number of qubits only growing logarithmically in $\varepsilon^{-1}$. As a corollary, we also obtain a universal approximation result for quantum extreme learning machines when the error is measured with respect to the $L^2$-norm. 
By providing quantitative universal approximation results, 
the obtained error bounds provide an important theoretical foundation for trainable and randomised quantum neural networks.
}

{An important characteristic of the obtained bounds is that the approximation rates do not deteriorate as the dimension of the input data increases. In contrast, methods---applicable to general continuous functions---based on quantum feature maps \cite{Goto2021} would generally require $\mathcal{O}(\varepsilon^{-1})$ qubits and $\mathcal{O}(\varepsilon^{-d})$ measurement basis functions 
(where~$d$ is the data dimension) to achieve {approximation error} $\varepsilon>0$.
In particular, for small {approximation error} $\varepsilon$ the number of required observables grows exponentially in the dimension $d$, that is, the \textit{curse of dimensionality} occurs for quantum feature maps. In contrast, the approximation error rate obtained here for quantum neural networks does not suffer from the curse of dimensionality {when employed to approximate continuous, integrable functions with integrable Fourier transform.} 

{The remainder of the article is structured as follows. In Section~\ref{sec:variational} we construct the universal variational quantum circuit and state the first quantum neural  network approximation error result. Section~\ref{sec:ReservoirQ} then introduces the modified circuit with randomly generated weights and proves an analogous result. In Section~\ref{sec:variationalInfty}, we complement these results by bounds with respect to the stronger $L^\infty$-error metric. Section~\ref{sec:conlcusion} concludes the article. All proofs are collected in Section~\ref{sec:proofs}.
}

\section{Approximation error bounds for variational quantum circuits}
\label{sec:variational}
{In this section we consider the most popular family of quantum neural networks: variational quantum circuits with trainable parameters. We start by introducing the considered variational quantum circuits. Subsequently, in Theorem~\ref{thm:Approx} we then obtain approximation error bounds for our circuits and target functions with integrable Fourier transform. Finally, in Corollary~\ref{cor:universality} we prove that the considered circuits are universal. For convenience, key notation for Section~\ref{sec:variational} is collected in Table~\ref{tab:symbols1}.}

\begin{table}[h!]
	\centering
	\renewcommand{\arraystretch}{1.3}
	\caption{Key notation for Section~\ref{sec:variational}}
	\label{tab:symbols1}
	{\color{black}
		\begin{tabular}{||c c c||} 
						\hline
			{	Symbol} & {Meaning} & {Definition} \\
			\hline
			$\ttheta$ & variational quantum circuit parameters & Section~\ref{subsec:circuit} \\ 
			$\xx$ & inputs & Section~\ref{subsec:circuit} \\ 
			$n$ & accuracy parameter & Section~\ref{subsec:circuit} \\
			$ \nq$ &  number of qubits & Definition~\ref{def:QC} \\ 
			$\Cg_{\nq}(\ttheta,\xx)$ & variational quantum circuit & Definition~\ref{def:QC} \\ 
			$\P_m^{n}$ & measured probabilities  & Definition~\ref{def:QC} 
			\\ $\ftrn$ & quantum neural network & \eqref{eq:qnn}
			\\ 
			$\LOnefhat$ & Fourier integral &\eqref{eq:Lonefhat}  \\
			$\Ff$ & space of continuous, integrable functions & \eqref{eq:Space1} \\
			$\Ff_R$ & functions in $\Ff$ with $\LOnefhat\leq R$ & \eqref{eq:Space1R}  \\ [1ex] 
			\hline
		\end{tabular}
	}
\end{table}

\subsection{Construction of a universal variational quantum circuit}
\label{subsec:circuit}
{In this section we construct the considered variational quantum circuits. For an accuracy parameter} $n \in \N$, 
weights $\ab = (\ab^{1},\ldots,\ab^{n}) \in (\RR^d)^n$, $\bb = (b^{1},\ldots,b^{n}) \in \R^n$, 
$\gb = (\gamma^{1},\ldots,\gamma^{n}) \in [0,2\pi]^n$ and an input $\xx=(x_1,\ldots,x_d) \in \RR^d$
define the gates $\Ug^{(i)}_1 := \Ug^{(i)}_1\left(\ab^{i},b^{i}, \xx\right)$ and $ \Ug^{(i)}_2  :=  \Ug^{(i)}_2 \left(\gamma^{i}\right)$ acting on a single qubit:
\begin{equation*}
\begin{array}{cl}
\Ug^{(i)}_1 
& := \Hg \, \Rgz\left(-b^{i}\right) \Rgz\left(-a^{i}_d x_d\right) \cdots \Rgz\left(-a^{i}_1 x_1\right) \Hg,
\\
\Ug^{(i)}_2 & := \Rgy\left(\gamma^{i}\right),
\end{array}
\end{equation*}
with~$\Hg$ the Hadamard gate and~$\Rgy$, $\Rgz$ the rotations around the $\mathrm{Y}$-and the $\mathrm{Z}$-axis respectively: 
$$
\Rgy(\gamma)
:= \begin{pmatrix} \cos\left(\frac{\gamma}{2}\right) & - \sin\left(\frac{\gamma}{2}\right) \\ \sin\left(\frac{\gamma}{2}\right) & \cos\left(\frac{\gamma}{2}\right) \end{pmatrix}
\text{,} \,\,
\Rgz(\alpha) := \begin{pmatrix} \E^{-\I \frac{\alpha}{2}} & 0 \\ 0 & \E^{\I \frac{\alpha}{2}} \end{pmatrix},
$$
for $\alpha\in\RR$ and $\gamma \in [0,2\pi]$.
Now write $\ttheta=(\ab^{(i)},b^{(i)},\gamma^{(i)})_{i=1,\ldots,n}
\in \TTheta := (\RR^d\times\RR\times[0,2\pi])^n$, $\bar{\Ug}^{(i)} = \Ug^{(i)}_1 \otimes \Ug^{(i)}_2 $ 
and define the block matrix $\Ug := \Ug(\ttheta,\xx)$ by 

\[
\Ug :=   \begin{bmatrix}
\bar{\Ug}^{(1)} & \mathbf{0}_{4 \times 4}  & \mathbf{0}_{4 \times 4} & \cdots &  \mathbf{0}_{4 \times 4} & \mathbf{0}_{4 \times n_0} \\
\mathbf{0}_{4 \times 4} & \bar{\Ug}^{(2)} & \mathbf{0}_{4 \times 4} & \cdots &  \mathbf{0}_{4 \times 4} & \vdots  \\
\vdots &  & \ddots &  & \vdots & \vdots \\
\mathbf{0}_{4 \times 4} & \cdots &  \mathbf{0}_{4 \times 4} &  \bar{\Ug}^{(n-1)} & \mathbf{0}_{4 \times 4} & \vdots 
\\
\mathbf{0}_{4 \times 4} & \cdots &  \cdots & \mathbf{0}_{4 \times 4} & \bar{\Ug}^{(n)} & \mathbf{0}_{4 \times n_0}
\\
{\bf 0}_{n_0 \times 4} & \cdots & \cdots & \cdots & {\bf 0}_{n_0 \times 4} & {\bf 1}_{n_0 \times n_0}
\end{bmatrix},
\]
with $n_0$ {such that the matrix dimension $4n+n_0$ is a power of~$2$}.\footnote{{Precisely, $n_0$ is the smallest integer in~$\N_0$ such that $\log_2(4n+n_0) \in \N$.}}  
Then $\Ug \in \C^{(4n+n_0) \times (4n+n_0)}$ is unitary and can be viewed as a gate operating on $\nq := \lceil \log_2(4n) \rceil $ qubits. Let $N=4n+n_0=2^{\nq}$ and $\Vg \in \C^{N \times N}$ any unitary matrix that maps 
$\ket{0}^{\otimes \nq}$ to the state $\ket{\psi} = \frac{1}{\sqrt{n}} \sum_{i=0}^{n-1} \ket{4i}$. {We refer to Appendix~\ref{sec:n0} for a more detailed discussion on the choice of~$n_0$ and the construction of $\ket{\psi}$.}

\begin{remark}
Since $\Vg\ket{0}^{\otimes\nq} = \ket{\psi}$, the first column of~$\Vg$ (seen as a matrix in $\CC^{N\times N}$) corresponds to~$\ket{\psi}$ (seen as a vector in $\CC^{N}$).
We provide an explicit (but not necessarily unique) construction for such a unitary matrix~$\Vg$
given this constraint on the first column. 
Indeed, set 
$\Vg := 2\ket{\varphi}\bra{\varphi} - \Ig$,
with 
$$
\ket{\varphi} := \frac{\ket{0}+\ket{\psi}}{\sqrt{2\left(1+\braket{0|\psi}\right)}},
$$
where we write $\ket{0}$ in place of $\ket{0}^{\otimes\nq}$ for brevity here.
In this case, immediate computations show that~$\Vg$ is unitary since 
$\Vg^\dagger := 2\ket{\varphi}\bra{\varphi} - \Ig = \Vg$
and
$\Vg\Vg^\dagger = \Vg^\dagger\Vg = \Ig$.
Furthermore,
\begin{align*}
\Vg\ket{0} & = \left(2\ket{\varphi}\bra{\varphi} - \Ig\right)\ket{0}\\
 & = \left(2\frac{\ket{0}+\ket{\psi}}{\sqrt{2\left(1+\braket{0|\psi}\right)}}\frac{\bra{0}+\bra{\psi}}{\sqrt{2\left(1+\braket{0|\psi}\right)}} - \Ig\right)\ket{0}\\
 & = \frac{\ket{0}\bra{0}+\ket{0}\bra{\psi}+\ket{\psi}\bra{0}+\ket{\psi}\bra{\psi}}{1+\braket{0|\psi}}
 \ket{0} - \ket{0}\\
 & = \frac{\ket{0}+\ket{0}\braket{\psi|0}+\ket{\psi}+\ket{\psi}\braket{\psi|0}}{1+\braket{0|\psi}} - \ket{0}\\ 
  & = \frac{\ket{0}\left(1+\braket{\psi|0}\right)+\ket{\psi}\left(1+\braket{\psi|0}\right)}{1+\braket{0|\psi}} - \ket{0}
  = \ket{\psi},
\end{align*}
where we used the fact that $\braket{0|\psi} \in \RR$ in the third line.
The property $\Vg\ket{0}^{\otimes\nq} = \ket{\psi}$ is in fact the only property of $\Vg$ that is required, otherwise  $\Vg$  does not have any impact on the scheme or the error bounds. In case there are several alternative choices for $\Vg${,}  one may thus select the one that is most suitable from the perspective of hardware requirements and limitations.
\end{remark}
We can now measure the state of the $\nq$-qubit system after applying the gates $\Vg$ and $\Ug$. 
The possible states that we could measure are ${0,\ldots,N-1}$. 
For $m \in \{0,1,2,3\}$ we denote by~$\P_{m}$ the probability that the measured state is in  $\{m,4+m,\ldots,4(n-1)+m\}$.
By running the circuit~$S$ times we may estimate these probabilities by
\begin{equation}\label{eq:probestimate}
\widehat{\P}_{m}^{n} := \frac{1}{S} \sum_{s=1}^S \mathbbm{1}_{\{m,4+m,\ldots,4(n-1)+m\}}(i^{(s)}),
\end{equation}
with $i^{(s)}$ the measured state in the $i$-th shot. 
{For conciseness} we will assume~$S$ large enough and will consider $\P_{m}^{n} = \P_{m}^{n}(\ttheta,\xx)$ as the output of the quantum circuit. {As shown in Appendix~\ref{appendixC} this is not a limitation; the Monte Carlo error can be easily taken into account, leading to an additional error of order  $1/\sqrt{S}$ for independent shots.}
We summarise this construction in the following formal definition:

\begin{figure}
\centering
\includegraphics[width=2.5in]{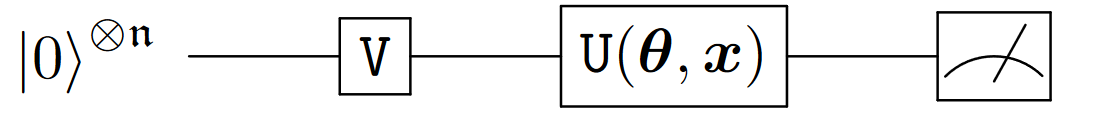}
\caption{{Abstract description of the quantum circuit in Definition~\ref{def:QC}. The initial state $\ket{0}^{\otimes \nq}$ is processed through the quantum gates~$\Vg$ and~$\Ug$ and then measured.}}
\label{fig:abstr}
\end{figure}

\begin{definition}\label{def:QC}
For $n \in \NN$ and $\ttheta \in\TTheta$, 
we define the unitary operator $\Cg_{\nq}(\ttheta,\xx) := \Ug(\ttheta,\xx)\Vg$ acting on $\nq = \lceil \log_2(4n) \rceil $ qubits, namely the realisation of the variational quantum circuit
{drawn abstractly in Fig.~\ref{fig:abstr}. The circuit  acts on the initial state $\ket{0}^{\otimes \nq}$ via the quantum gates $\Vg$ and $\Ug$ and then measures the state.}
 Furthermore, for any $m \in \{0,1,2,3\}$, we let
 $\P_m^{n} := \PP\left("\Cg_{\nq}(\ttheta,\xx)\ket{0}^{\otimes \nq} \in \{m,4+m,\ldots,4(n-1)+m\}"\right)$.
\end{definition}

\begin{remark}\label{rmk:alternative}
{Instead of the above quantum circuit (used in Theorem~\ref{thm:Approx}), we could also employ an alternative circuit with $\nq :=2n$ qubits, creating an entangled tensor structure. More specifically, the alternative circuit maps
$\ket{0}^{\otimes 2n}$ to the state 
$\frac{1}{\sqrt{n}} \sum_{i=0}^{n-1} (\ket{1} \otimes \ket{1} )^{\otimes i} (\Ug_1^{(i+1)}\ket{0} \otimes \Ug_2^{(i+1)} \ket{0} ) \otimes (\ket{0} \otimes \ket{0} )^{\otimes n-i-1}$. Choosing $\P_m^{n}$ ($m\in \{0,1,2,3\}$) as the probability that the measured state is in 
$\{ (\ket{1} \otimes \ket{1} )^{\otimes i} \otimes (\ket{x} \otimes \ket{y} ) \otimes (\ket{0} \otimes \ket{0} )^{\otimes n-i-1} \,: \, i \in \{0,\ldots,n-1\}\}$, where $x,y \in \{0,1\}$ are the coefficients in the binary representation $m=2x+y$, the proof of Proposition~\ref{prop:circuitOutput} below carries over analogously. A detailed explicit construction for building such a circuit is left for future study. 
}
\end{remark}

\subsection{Approximation error bound}
We now consider learning based on the variational quantum circuit constructed above. 
For a continuous and integrable (in~$\Cc(\RR^d)\times L^{1}(\RR^d)$) function $f \colon \RR^d \to \R$, 
we denote by $\widehat{f}(\bx) := \int_{\RR^d} \E^{-2\pi\I \boldsymbol{y} \cdot \bx} f(\boldsymbol{y}) \D \boldsymbol{y}$, for $\bx \in \RR^d$, 
its Fourier transform{,} and define
\begin{equation}\label{eq:Lonefhat}
\LOnefhat :=\int_{\RR^d} |\widehat{f}(\bx)|\D \bx,
\end{equation}
which may or may not be finite.

Given the operator $\Cg_{\nq}(\ttheta,\xx)$
and $R>0$,
introduce the map $\ftrn:\RR^d \to\RR$ by
\begin{equation}\label{eq:qnn}
\ftrn(\cdot)
:= R-2R[\P_1^{n}(\ttheta,\cdot)
+ \P_2^{n}(\ttheta,\cdot)].
\end{equation}
{Fig.~\ref{fig:1} provides a schematic diagram of how the quantum neural network $\ftrn$ acts on inputs $\xx$ through the quantum circuit: the initial state $\ket{0}^{\otimes \nq}$ is processed through the quantum gates~$\Vg$ and $\Ug(\ttheta,\xx)$ and then measured. The input  $\xx$ determines the operator $\Ug(\ttheta,\xx)$ and the probabilities $\P_1^{n}(\ttheta,\xx)$, $\P_2^{n}(\ttheta,\xx)$. These probabilities are aggregated into the network output $\ftrn$ according to \eqref{eq:qnn}.}

\begin{figure}
\centering
\includegraphics[width=2.5in]{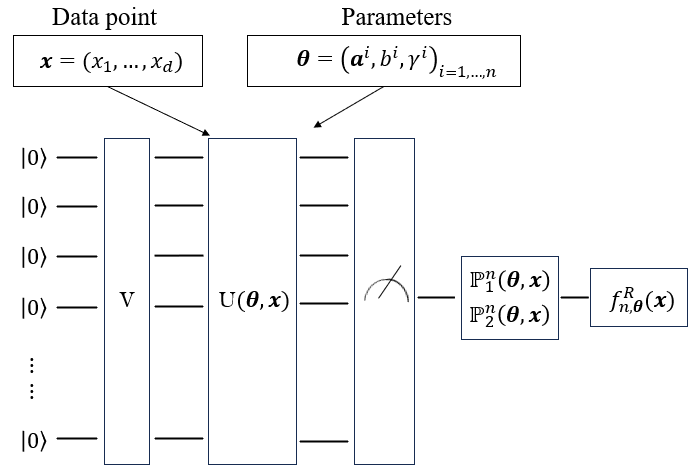}
\caption{{Schematic diagram of the quantum neural network $\ftrn$ with $\nq = \lceil \log_2(4n) \rceil $ qubits. The initial state $\ket{0}^{\otimes \nq}$ is processed through the quantum gates~$\Vg$ and~$\Ug$ and then measured; yielding probabilities $\P_1^{n}$ and $\P_2^{n}$. These probabilities are aggregated into the network output $\ftrn$ according to \eqref{eq:qnn}.}}
\label{fig:1}
\end{figure}

In order to avoid repetitions later and to compare the different results more clearly, introduce the {space $\Ff$ of continuous, integrable functions and its subspace $\Ff_R$ of functions with Fourier integral at most $R$:} 
\begin{align}
\Ff & := \Big\{f:\RR^d\to\RR: f\in \Cc\left(\RR^d\right)\cap L^{1}\left(\RR^d\right)\Big\},\label{eq:Space1}\\
\Ff_{R} & := \left\{f\in\Ff,
\text{ with }\LOnefhat \leq R\right\},
\qquad\text{for any }R>0.\label{eq:Space1R}
\end{align}

Before stating the first theorem, we fix a probability measure~$\mu$ on~$\RR^d$,
which shall be used throughout to measure the approximation error. 
This measure can be chosen arbitrarily and changes how the approximation quality is weighted in different regions of~$\RR^d$. 
For example, if~$\mu$ is the uniform measure on a hypercube $[-M,M]^d$ (for some $M>0$), then the approximation quality is weighted equally over $[-M,M]^d$ and no weight is put on the approximation quality outside $[-M,M]^d$.
The following result shows that the outputs of~$\Cg_{\nq}(\ttheta,\xx)$ can be used to approximate such  functions---further satisfying $\widehat{f} \in L^{1}(\RR^d)$---up to an error of size $n^{-\frac{1}{2}}$. 
In particular, no curse of dimensionality occurs and the number of qubits is logarithmic in~$n$.  

\begin{theorem}\label{thm:Approx}
For any $R>0$, $f \in \Ff_{R}$ and
$n\in\NN$, there exists~$\ttheta\in\TTheta$ such that 
\begin{equation}
\left(\int_{\RR^d} 
 \left|f(\xx) - \ftrn(\xx)\right|^2 \mu(\D \xx)\right)^{1/2} \leq  \frac{\LOnefhat}{\sqrt{n}}.
\end{equation}
\end{theorem}

\begin{remark}
We would like to emphasise that the quantum circuit above depends on~$\ttheta$ (a quantum analogue of the vector of hyperparameters in classical neural networks) and the input data~$\xx$ {as, for example,} in~\cite{perez2021one, Schuld2021}.
Theorem~\ref{thm:Approx} states that the optimal vector~$\ttheta$ does not depend on~$\xx$.
{For example, }
{fix a weighting function $w \colon \R^d \to [0,\infty)$ that models the weight $w(\xx)$ we assign to the error at each point $\xx \in \R^d$.}
Provided that $w \in L^1(\R^d)$, we can use it to define a probability measure~$\mu$ with Lebesgue-density $w/\|w\|_1$. In this case the error {in Theorem~\ref{thm:Approx} is given by}
$$
\begin{aligned}
\int_{\RR^d} &  
\left|f(\xx) - \ftrn(\xx)\right|^2 \mu(\D \xx)
\\ & = \frac{1}{\|w\|_1} \int_{\RR^d} w(\xx)
\left|f(\xx) - \ftrn(\xx)\right|^2 \D \xx.
\end{aligned}
$$
{Hence,} the quantum neural network is able to approximate the function $f$ on all of~$\R^d$ in an average sense with a single vector of parameters~$\ttheta$. 
Handling or adding constraints to take care of potential erratic behaviour outside the (necessarily) compact set of training data points has recently been tackled for classical neural networks~\cite{balestriero2021learning,siegel2022parallel,xu2020neural}{,} and we leave it to future research in the quantum case.
\end{remark}

{
\begin{remark}
To put the obtained bounds in context{,} let us compare them to bounds available for the related method of quantum feature maps.
In classical machine learning, random feature methods  are closely linked to kernel methods, as in~\cite{RahimiRecht2008a, RahimiRecht2008} and in the more recent works~\cite{MMM21, RudiRosasco2017}. Similarly, quantum reservoir networks are related to quantum feature maps as studied{,} for example{,} in~\cite{Goto2021}, {proving universal approximation results} for quantum feature maps. A collection of basis functions is built by applying multiple observables to a quantum circuit. Linear combinations of these basis functions are then shown to be universal approximators. The proof is based on polynomial approximations, which also yields error bounds for these quantum feature maps. These results are valid for all Lipschitz functions and require $\mathcal{O}(\varepsilon^{-1})$ qubits and $\mathcal{O}(\varepsilon^{-d})$ measurement basis functions 
(where~$d$ is the data dimension) to achieve {approximation error} $\varepsilon>0$. In particular, for small $\varepsilon$ the number of required observables grows exponentially in the dimension $d$, i.e., the \textit{curse of dimensionality} occurs. In contrast, here we consider a single quantum circuit,  linearly combine probabilities obtained from a single measurement of this circuit and the approximation error rate does not suffer from the curse of dimensionality.
\end{remark}
}

For standard neural networks, the condition $\int_{\RR^d} \|\bx\| |\widehat{f}(\bx)|\D \bx < \infty$ guarantees that~$f$ can be approximated by neural networks without the curse of dimensionality, as established in~\cite{Barron1993}. 
For $f \in L^1(\RR^d)$, this condition is stronger than the requirement $\LOnefhat < \infty$ imposed in Theorem~\ref{thm:Approx}. 

A sufficient condition for $\LOnefhat < \infty$ is the requirement that $f \in \Hh^{s}(\RR^d)$ for $s>\frac{d}{2}$, 
where~$\Hh^{s}$ is the Sobolev space of order~$s$, see \cite[Lemma~6.5]{Folland1995} applied with $k=0$ and its proof. 
In dimension $d=1$, for example, a sufficient condition to apply Theorem~\ref{thm:Approx} is thus $f \in L^1(\R)\cap L^2(\R) \cap \Cc(\R)$ ($=L^1(\R) \cap \Cc(\R) $), $f' \in L^2(\R)$.

The proof of Theorem~\ref{thm:Approx} is constructive and consists of two steps. 
First, in Proposition~\ref{prop:circuitOutput}, 
we show that for any choice of weights $\ttheta=(\ab^{i},b^{i},\gamma^{i})_{i=1,\ldots,n}$ the function 
\begin{equation}
\begin{aligned}
	\label{eq:fct}
{\gtrn(\xx) = \frac{1}{n}\sum_{i=1}^{n} R\cos\left(\gamma^{i}\right) \cos\left( b^{i}+\ab^{i} \cdot \xx \right),}
\end{aligned}
\end{equation}
can be realised as the output of~$\Cg_{\nq}(\ttheta,\xx)$, namely that it is equal to~$\ftrn$. 
Then, in Proposition~\ref{prop:FourierApprox}, 
we use a probabilistic argument to show that functions of type~\eqref{eq:fct} are able to approximate continuous, integrable functions $f$ with integrable Fourier transform up to an error of size $\LOnefhat n^{-\frac{1}{2}}$,
thus yielding Theorem~\ref{thm:Approx}. {We refer to Section~\ref{subsec:proof1} for details.}

\subsection{Universal approximation by variational quantum circuits}

As a corollary of Theorem~\ref{thm:Approx} we obtain the following universal approximation result:

\begin{corollary} \label{cor:universality}
Let $\mu$ be a probability measure on $\RR^d$ and let $f \in L^2(\RR^d,\mu)$. 
Then for any $\varepsilon >0$ there exist $n \in \N$, $R>0$ and $\ttheta\in\TTheta$ such that 
$\Cg_{\nq}(\ttheta,\xx)$ outputs  $f_{\ttheta}(\xx)$ with 
\begin{equation}\label{eq:universality}
\left(\int_{\RR^d} |f(\xx)-f_{\ttheta}(\xx)|^2 \mu(\D\xx)\right)^{1/2} \leq  \varepsilon.
\end{equation}
\end{corollary}
{The proof of this result is provided in Section~\ref{subsec:proof2}. Corollary~\ref{cor:universality} shows  that our quantum circuit is universal in a mean-square sense, namely that it can approximate arbitrarily well any (not necessarily continuous) function~$f$ which is square integrable with respect to a probability measure~$\mu$. 
A key feature of this universality is that when~$\mu$ is chosen with full support (for example when it admits a strictly positive density), 
then the quantum neural network approximation holds on all of~$\R^d$. This makes the universality result particularly useful when dealing with stochastic inputs, 
as is common in applications to time series and finance. 
}
{\subsection{Comparison to classical neural networks}}
{In this section we provide a more detailed comparison of quantum neural networks to their classical counterparts.  Corollary~\ref{cor:universality}  establishes a mean-square universal approximation result for quantum neural networks. For classical neural networks such a result was established in \cite{hornik1991}. This indicates that, qualitatively, the two learning systems possess the same approximation capabilities. Moreover, Theorem~\ref{thm:Approx} proves that quantum neural networks are able to achieve an approximation error decaying as $n^{-1/2}$ for a larger class of $f \in L^1(\RR^d)$ than classical neural networks. Indeed, for classical neural networks the seminal work by Barron~\cite{Barron1992,Barron1993,Barron1994ApproximationAE} proved such a rate for functions $f$ satisfying $\int_{\RR^d} \|\bx\| |\widehat{f}(\bx)|\D \bx < \infty$.
For $f \in L^1(\RR^d)$, this condition is stronger than the requirement $\LOnefhat < \infty$ imposed in Theorem~\ref{thm:Approx}. 	
Thus, this comparison indicates that quantum neural networks possess a higher expressive power than classical neural networks.  In terms of learning efficiency, a general comparative advantage has yet to be established in the literature and may be highly dependent on the considered data, as demonstrated in \cite{PowerOfData}.
}

{Classical neural networks are typically trained using stochastic gradient descent-type algorithms, with loss gradients computed via backpropagation. For quantum neural networks,  analogously to classical neural networks, one may define a loss function, say, in the context of Corollary~\ref{cor:universality}, 
$$
\mathcal{E}(\theta):= \int_{\RR^d} |f(\xx)-f_{\ttheta}(\xx)|^2 \mu(\D\xx),
$$
and solve (numerically) the minimisation problem
$\min_{\theta}\mathcal{E}(\theta)$.\footnote{{Although the construction of the optimal parameter~$\ttheta$ is explicit in the proof of Theorem~\ref{thm:Approx}, it however depends on the (typically unknown) function~$f$ itself. 
It may thus not be possible to construct without knowing~$f$.}}
This falls into the scope of \emph{hybrid quantum-classical} systems, whereby the quantum circuit generates possible outputs~$f_{\theta}(x)$, but the minimisation algorithm runs classically, for example with gradient descent.
Classical neural network training may be sensitive to appropriate hyperparameter selection; allowing to avoid local minima of the loss surface. Quantum neural network training may exhibit challenges in the stability of the learning process due to Barren plateaus \cite{Barren}, which may lead to vanishing gradient problems for large number of qubits. }
The quantum \emph{parameter-shift rule} proposed in~\cite{mitarai2018quantum} is an interesting alternative that allows one to compute gradients purely with quantum circuits (i.e. without requiring a hybrid mode). 
However, its computational efficiency is not guaranteed as it requires additional quantum circuits, 
and furthermore {so far it has only been established for}
$f(\tilde{\theta}) := \bra{\psi(\tilde{\theta})}\mathcal{M}\ket{\psi(\tilde{\theta})}$ for some parameterised quantum state $\ket{\psi(\tilde{\theta})}$ and some observable~$\mathcal{M}$. {This is not our setup here, 
and we leave it to future work to possibly rewrite our result for such parameterisations.}

\section{Variational quantum circuits with randomly generated weights}\label{sec:ReservoirQ}

We now construct a quantum circuit in the spirit of reservoir computing: the parameters inside the circuit are randomly generated and only a final post-measurement layer of weights is trainable. 
Training such a circuit thus only requires to solve a linear regression problem, as opposed to generic quantum circuits, typically trained with gradient-based methods. 

{The approximation bounds we obtain for these quantum reservoirs are akin to the approximation bounds available for classical random feature neural networks in~\cite{Gonon2021, RC12}. 
As in these, we consider here functions whose Fourier transforms satisfy certain integrability conditions. We prove that these functions can be approximated by quantum reservoir networks up to an approximation error of order $n^{-\frac{1}{2}}$, where~$n$ is a parameter we may choose, proportional to the number of randomly generated weights. 
Put differently, our result shows that an approximation {error} $\varepsilon>0$ can be achieved by using $\mathcal{O}(\varepsilon^{-2})$ randomly generated weights and a circuit with $\mathcal{O} (\lceil \log_2(\varepsilon^{-1}) \rceil)$ qubits. 
In particular, these results provide a bound on the number of qubits and the size of the quantum circuit that is required in order to guarantee a prescribed approximation accuracy. Furthermore, we also obtain a universal approximation result, proving that any square-integrable functions can be approximated arbitrarily well by the constructed quantum random circuits. 
Our results provide mathematical foundation for the use of quantum extreme learning machines in practice.}
{For convenience, key notations for this section are collected in Table~\ref{tab:symbols2}.}
\vspace{-2mm}
\begin{table}[h!]
	\centering
	\renewcommand{\arraystretch}{1.3}
	\caption{Key notation for Section~\ref{sec:ReservoirQ}}
	\label{tab:symbols2}
	{\color{black}
		\begin{tabular}{||c c c||} 
			\hline
			{	Symbol} & {Meaning} & {Definition} \\
			\hline
			$\pi_a$ & density of random circuit parameters & Section~\ref{subsec:reservoirCircuit} \\ 
			$\xx$ & inputs & Section~\ref{subsec:reservoirCircuit} \\ 
			$n$ & accuracy parameter & Section~\ref{subsec:reservoirCircuit} \\
			$\nqb$ & number of qubits & Section~\ref{subsec:reservoirCircuit} \\
			$\overline{\Cg}_{\nqb}(\xx)$ & random quantum circuit & Section~\ref{subsec:reservoirCircuit}  \\ 
			$\PPb_k(\xx)$ & measured probabilities  & Section~\ref{subsec:reservoirCircuit} 
			\\ $F_{\ww}$ & quantum random neural network & \eqref{eq:randomF}
			\\ 
			$\Lbarf$ & $\pi_a$-weighted Fourier integral &\eqref{eq:Lbarf}  \\
			$\Ffov$ & space of functions``compatible'' with $\pi_a$ & \eqref{eq:Space1Bar}  \\ [1ex] 
			\hline
		\end{tabular}
	}
\end{table}

\subsection{{Construction of} a random universal quantum circuit}
\label{subsec:reservoirCircuit}

Let $n \in \N$,
$\Bf=(B^{i})_{i=1,\ldots,n}$ i.i.d.\ with $\half$-Bernoulli distribution,
$\Af=(\Af^{i})_{i=1,\ldots,n}$ i.i.d.\ with density~$\pi_a$, and~$\Af$ and~$\Bf$ independent,
and $\Af^i = (A^i_j)_{1\leq j \leq d}$ for each $i=1,\ldots, n$.
For an input $\xx \in \RR^d$ we consider,
for $i=1,\ldots, n$, $\overline{\Ug}^{(i)}_1:=\overline{\Ug}^{(i)}_1(\xx)
:= \Ug^{(i)}_1\left(2\pi \mathbf{A}^{i},\frac{\pi}{2} B^{i},\xx\right)$, that is, 
$$
\overline{\Ug}^{(i)}_1(\xx)
 = \Hg \, \Rgz\left(-B^{i}\frac{\pi}{2}\right) \Rgz\left(-2\pi {A}^{i}_d x_d\right) \cdots 
 \Rgz\left(-2 \pi {A}^{i}_1 x_1\right) \Hg.
$$
Similarly to Section~\ref{subsec:circuit} we now use these gates to build the (random) block matrix $\overline{\Ug} := \overline{\Ug}(\xx)$, 
\[
\overline{\Ug} :=
\begin{bmatrix}
	\overline{\Ug}^{(1)}_1 & \mathbf{0}_{2 \times 2}  & \mathbf{0}_{2 \times 2} & \cdots &  \mathbf{0}_{2 \times 2} & \mathbf{0}_{2 \times \overline{n}_0} \\
	\mathbf{0}_{2 \times 2} & \overline{\Ug}^{(2)}_1 & \mathbf{0}_{2 \times 2} & \cdots &  \mathbf{0}_{2 \times 2} & \vdots  \\
	\vdots &  & \ddots &  & \vdots & \vdots \\
	\mathbf{0}_{2 \times 2} & \cdots &  \mathbf{0}_{2 \times 2} &  \overline{\Ug}^{(n-1)}_1 & \mathbf{0}_{2 \times 2} & \vdots 
	\\
	\mathbf{0}_{2 \times 2} & \cdots &  \cdots & \mathbf{0}_{2 \times 2} & \overline{\Ug}^{(n)}_1 & \mathbf{0}_{n \times \overline{n}_0}
	\\
	{\bf 0}_{\overline{n}_0 \times 2} & \cdots & \cdots & \cdots & {\bf 0}_{\overline{n}_0 \times 2} & {\bf 1}_{\overline{n}_0 \times \overline{n}_0}
\end{bmatrix},
\]
with $\overline{n}_0 \in \N_0$ the smallest positive integer such that $\log_2(2n+\overline{n}_0) \in \N$.  
Then $\overline{\Ug} \in \C^{(2n+\overline{n}_0) \times (2n+\overline{n}_0)}$ is a (random) unitary matrix and can be viewed as a gate operating on $\nqb := \lceil \log_2(2n) \rceil $ qubits. Let $\overline{N} := 2n+\overline{n}_0=2^{\nqb}$ and similarly to Section~\ref{subsec:circuit} let $\overline{\Vg} \in \C^{\overline{N} \times \overline{N}}$ be any unitary matrix that maps 
$\ket{0}^{\otimes \nqb}$ to the state $\ket{\overline{\psi}} = \frac{1}{\sqrt{n}} \sum_{i=0}^{n-1} \ket{2i}$. 

Consider now the circuit {processing the initial state $\ket{0}^{\otimes \nqb}$ through the quantum gates $\overline{\Vg}$ and $\overline{\Ug}(\xx)$ and then measuring the state. We call this circuit $\overline{\Cg}_{\nqb}(\xx)$. Figure~\ref{fig:abstr} (with~$\Vg$ and~$\Ug$ replaced by $\overline{\Vg}$ and $\overline{\Ug}$) provides an abstract depiction of the circuit. }
The possible measurement outcomes are $0,\ldots,\overline{N}-1$,
and we denote $\PPb_k := \PPb_k(\xx)$ the probability that the measured state is equal to 
$k \in \{0,\ldots,\overline{N}-1\}$,
 which can be estimated by running the circuit, 
 as explained in Section~\ref{subsec:circuit}.
In contrast to Section~\ref{sec:variational}, here no parameter is trained/adjusted within the quantum circuit. 
The matrix parameters~$\Af$ and~$\Bf$ are randomly generated and then fixed, 
and  the subsequent inputs~$\xx$ are mapped through the fixed circuit to output probabilities 
$\PPb_k = \PPb_k(\xx) = \PPb_k^{\Af,\Bf}(\xx)$, for $k=0,\ldots,\overline{N}-1$. 

{
\begin{remark}
{Instead of the quantum circuit introduced here and used in Theorem~\ref{thm:random}  below, one could alternatively use a circuit with $\nqb :=n$ qubits creating an entangled tensor structure. More specifically, the circuit could be chosen to map
$\ket{0}^{\otimes n}$ to the state 
$\frac{1}{\sqrt{n}} \sum_{i=0}^{n-1} \ket{1}^{\otimes i} \otimes \overline{\Ug}^{(i+1)}_1 \ket{0} \otimes  \ket{0}^{\otimes n-i-1}$. 
Now we choose for $j \in \{0,\ldots,n-1\}$ the probability $\PPb_{2j} := \PPb_{2j}(\xx)$ in \eqref{eq:randomF} below as the probability that the measured state is 
$\ket{1}^{\otimes j} \otimes \ket{0}^{\otimes n-j}$. Then,
the proof of Proposition~\ref{prop:circuitOutputRandom} below carries over analogously also for this modified circuit. }
\end{remark}
}

\begin{remark}
In quantum reservoir computing there is usually a dynamical aspect, whereby the inputs that need to be processed are sequences. Here we work in a static setting, but still use the term \textit{reservoir quantum circuit} to emphasise that the parameters within the variational quantum circuits are randomly generated and then fixed. Alternatively, we also use the terminology \textit{quantum extreme learning machine} or \textit{quantum random feature network} in analogy to the classical terminology in the static case. 
\end{remark}

\begin{remark}
In the notation {of Section~\ref{sec:variational}}, we could consider a circuit in which the weights $(b^{i},\ab^{i})_{i=1,\ldots,n}$ are randomly sampled and then fixed, and only the remaining weights $(\gamma^{i})_{i=1,\ldots,n}$ are trained. 
Here we consider an even simpler circuit which does not depend on $\gamma^{1},\ldots,\gamma^{(n)}$, but instead on some trainable ``readout weights'' appearing only after measurement, {see \eqref{eq:randomF} below.}
\end{remark}

\subsection{{Approximation error bounds}}

Mimicking the previous section, we introduce {a subspace of  $\Ff$ that consists of functions with Fourier transform ``compatible'' with $\pi_a$: }
\begin{equation}\label{eq:Space1Bar}
\begin{aligned}
\Ffov := & \Big\{f:\RR^d\to\RR: f \in \Cc(\RR^d) \cap L^1(\RR^d), \widehat{f}\in L^1(\RR^d), \\ & \quad |\widehat{f}|\ll \pi_a, 
\int_{\RR^d} \frac{|\widehat{f}(\bx)|^2}{\pi_a(\bx)} \D\bx<\infty\Big\}.
\end{aligned}
\end{equation}
{Here} we write $\nu\ll \mu$ if $\mu(A) = 0$ implies $\nu(A) = 0$ ({i.e., the} measure~$\nu$ is  absolutely continuous with respect to the measure~$\mu$). {Moreover,} for a function $f\in\Ffov$, we denote
\begin{equation}\label{eq:Lbarf}
\Lbarf := \left(2\int_{\RR^d} \frac{|\widehat{f}(\bx)|^2}{\pi_a(\bx)} \D\bx\right)^{1/2}.
\end{equation}

Given the {probabilities $\PPb_k(\xx)${,} measured after running the circuit defined by the operator $\overline{\Cg}_{\nqb}(\xx)$}
and $\ww\in\RR^n$, introduce the map $F_{\ww}:\RR^d \to\RR$ by
\begin{equation}\label{eq:randomF}
F_{\ww}(\xx) := \sum_{j=0}^{n-1} w_j \left(2\PPb_{2j}(\xx)-\frac{1}{n}\right).
\end{equation}
{Fig.~\ref{fig:4} provides a schematic diagram of how the quantum random neural network $F_{\ww}$ acts on inputs $\xx$ through the quantum circuit: the initial state $\ket{0}^{\otimes \nqb}$ is processed through the quantum gates~$\overline{\Vg}$ and~$\overline{\Ug}(\xx)$ and then measured. The input  $\xx$ determines the operator $\overline{\Ug}(\xx)$ and the probabilities $\PPb_{0}(\xx), \ldots,\PPb_{2(n-1)}(\xx)$. These probabilities are aggregated into the network output~$F_{\ww}$ according to~\eqref{eq:randomF}}.
\begin{figure}
\centering
\includegraphics[width=2.8in]{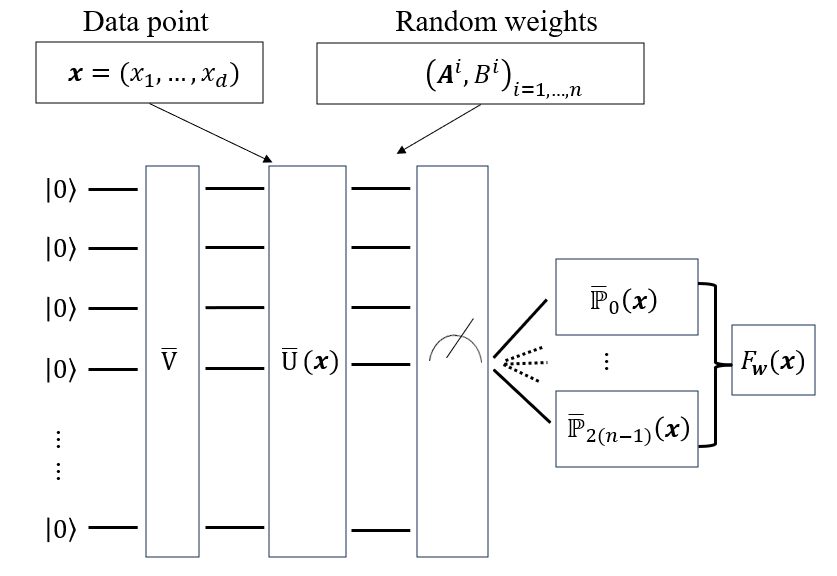}
\caption{{Schematic diagram of the quantum random neural network $F_{\ww}$ with $\nqb = \lceil \log_2(2n) \rceil $ qubits. The initial state $\ket{0}^{\otimes \nqb}$ is processed through the quantum gates~$\overline{\Vg}$ and~$\overline{\Ug}$ and then measured, yielding probabilities $\PPb_0, \ldots, \PPb_{2(n-1)}$. These probabilities are aggregated into the network output~$F_{\ww}$ according to~\eqref{eq:randomF}.}}
\label{fig:4}
\end{figure}
The following theorem provides an approximation error bound for such a circuit, which can be viewed as a quantum analogue to the bounds obtained for classical random feature neural networks in \cite{Gonon2021, RC12}.

\begin{theorem}\label{thm:random}
For any $n \in \N$, $f\in\Ffov$,
 there exists an $\RR^n$-valued $\sigma(\Af,\Bf)$-measurable random vector~$\Wf$ such that 
$$
\EE\left[\int_{\RR^d} \left|F_{\Wf}(\xx)-f(\xx)\right|^2 \mu(\D\xx)\right]^{1/2} \leq \frac{\Lbarf}{\sqrt{n}}.
$$
\end{theorem}

The proof of Theorem~\ref{thm:random} is quite similar to that of Theorem~\ref{thm:Approx}, it follows by combining Propositions~\ref{prop:FourierApproxRandom} and~\ref{prop:circuitOutputRandom} below. {We refer to Section~\ref{subsec:proof3} below for details.}

\begin{example}
As an example, consider $d=1$ and let~$\pi_a$ be the density of a~$t_1$-distribution. 
Then the condition 
$\Lbarf < \infty$ translates to $\int_{\RR} (1+|\bx|^2)|\widehat{f}(\bx)|^2 \D\bx < \infty$, namely $f \in \Hh^1(\R)$. 
Thus{,} $f$ satisfies the conditions of Proposition~\ref{prop:FourierApproxRandom} if it has a weak derivative~$f'$ and if both~$f$ and~$f'$ are square-integrable. 
	
More generally, for $d \in \N$ and~$\pi_a$ the density of a $t_{\nu}(0,\mathbbm{1}_d)$-distribution (with $\nu>0$),
the condition $\Lbarf < \infty$ is equivalent to 
$\int_{\RR^d} |\widehat{f}(\bx)|^2(1+\|\bx\|^2)^{(\nu+d)/2} \D\bx < \infty$, 
which means that~$f$ needs to belong to the Sobolev space $\Hh^{\frac{\nu+d}{2}}(\RR^d)$.
\end{example}

The next result shows that in the case where $f \in \Hh^{s}(\RR^d)$ for some $s>\frac{d}{2}$ we may in fact generate~$\mathbf{A}^{i}$ from a distribution arbitrary close to any given  density, at the expense of a possibly large constant in the error bound. {The proof is given in Section~\ref{subsec:proof4}.}
\begin{corollary} \label{cor:2}
Fix $s>\frac{d}{2}$.
Let $\varphi$ be an arbitrary probability density, 
$\nu = 2(s-\frac{d}{2})$ and~$t_\nu$ the density of a  $t_{\nu}(0,\mathbbm{1}_d)$-distribution. 
For $\delta \in (0,1)$, set $\pi_a := \delta t_\nu + (1-\delta) \varphi$. 
Let $n \in \N$ and let $F_w$  be the function realised by the quantum circuit in Theorem~\ref{thm:random} with weights~$\Af^i$ distributed according to~$\pi_a$ and $B^{i}$ distributed according to a $\half$-Bernoulli distribution. 
	Then for any $f \in \Hh^{s}(\RR^d) \cap L^1(\RR^d)$
	there exists a $\sigma(\Af,\Bf)$-measurable random vector~$\Wf$ and a constant $C_{f,\delta}$ independent of~$n$ such that 
	\[
	\EE\left[\int_{\RR^d} \left|F_{\Wf}(\xx)-f(\xx)\right|^2 \mu(\D\xx)\right]^{1/2} \leq \frac{C_{f,\delta}}{\sqrt{n}}.
	\]
\end{corollary}

\subsection{{Universal approximation by random quantum neural networks}}

{As a corollary of Theorem~\ref{thm:random}{,} we obtain the following universal approximation result:
\begin{corollary} \label{cor:universalityRandom}
Let $\mu$ be a probability measure on $\RR^d$,  let $f \in L^2(\RR^d,\mu)$ and assume that $\pi_a$ is continuous, strictly positive and decays at most polynomially. 
Then for any $\varepsilon >0$ there exist $n \in \N$ and an $\RR^n$-valued $\sigma(\Af,\Bf)$-measurable random vector~$\Wf$
such that  
\begin{equation}\label{eq:universalityRandom}
\EE\left[\int_{\RR^d} \left|F_{\Wf}(\xx)-f(\xx)\right|^2 \mu(\D\xx)\right]^{1/2} \leq  \varepsilon.
\end{equation}
\end{corollary}}
{The proof of this result follows completely analogously to  Corollary~\ref{cor:universality}, using the fact that{,} for a Schwartz function~$g${,} the integral $\overline{L}_{2}[g]$ is finite when $\pi_a$ is continuous, strictly positive and decays at most polynomially.
Corollary~\ref{cor:universalityRandom} shows that also the constructed random quantum circuit is universal in a mean-square sense: any (not necessarily continuous) function~$f$ which is square integrable with respect to a probability measure~$\mu$ can be approximated arbitrarily well using the considered random quantum circuits. 
}

\section{$L^\infty$-error bounds for variational quantum circuits}
\label{sec:variationalInfty}
We now prove approximation error bounds for variational quantum circuits in the case where the error is measured with respect to the uniform norm on compacts. {The section contains results for both trainable and randomized variational quantum circuits. We start with a result for the former situation. }

\subsection{$L^\infty$-error bounds for trainable variational quantum circuits}

{We start with the case of the quantum circuit introduced in Section~\ref{subsec:circuit}.}
For $R>0$, $\Ff_{R}$ as in~\eqref{eq:Space1R} and
$f \in \Ff_{R}$, let
\begin{equation}
\begin{aligned}
& \|f\|_{\mathcal{B}_2} := \left(\int_{\RR^d} \|\bx\|^2 |\widehat{f}(\bx)|\D \bx\right)^{\frac{1}{2}},
\\ & 
\qquad\text{and}\qquad
\mathcal{B}_R := \{f \in \Ff_{R} \, \colon \, \|f\|_{\mathcal{B}_2} < \infty  \}
\end{aligned}
\end{equation}
be a subset of $\Ff_{R}$ with further integrability properties of the Fourier transform. For functions in $\mathcal{B}_R$ we can complement the $L^2(\R^d,\mu)$-error bound in Theorem~\ref{thm:Approx} by a uniform error bound on compact sets.
\begin{theorem}\label{thm:ApproxUniform} 
	For any $R,M>0$, $f \in \mathcal{B}_{R}$ and
	$n\in\NN$, there exists~$\ttheta\in\TTheta$ such that 
	\begin{equation}\label{eq:uniformApproxError} \begin{aligned}
	&\sup_{\xx \in [-M,M]^d} 
	\left|f(\xx) - \ftrn(\xx)\right| \\ & \leq  \frac{2 (\pi+1) L^{1}[\widehat{f}] 
		+ 8 \pi  M d^{\frac{1}{2}}
		\LOnefhat^{\frac{1}{2}}\|f\|_{\mathcal{B}_2}}{\sqrt{n}}.
\end{aligned}	\end{equation}
\end{theorem}

{The proof of Theorem~\ref{thm:ApproxUniform} is given in Section~\ref{subsec:proof5}. The function class to which Theorem~\ref{thm:ApproxUniform}  is applicable is more restrictive than the one considered in Theorem~\ref{thm:Approx}:  functions in $\Ff_{R}$ possess a bounded Fourier transform, hence any function with $\|f\|_{\mathcal{B}_2}< \infty$ also satisfies $\LOnefhat<\infty$. 
On the other hand, the approximation error bound \eqref{eq:uniformApproxError} holds uniformly for each $x \in [-M,M]^d$ and the approximation error is also of order $\mathcal{O}(\frac{1}{\sqrt{n}})$.}

\subsection{$L^\infty$-universal approximation theorem}
As a corollary to Theorem~\ref{thm:ApproxUniform}{,} we obtain a universal approximation for continuous functions on compact subsets of $\R^d$ and with error measured with respect to the uniform norm. {This complements the universal approximation result in the $L^2$-sense proved above in Corollary~\ref{cor:universality}.}
\begin{corollary}\label{cor:3}
Let $\mathcal{X} \subset \RR^d$ be compact and $f \in \Cc(\mathcal{X},\RR)$.	Then for any $\varepsilon >0$ there exist $n \in \N$, $R>0$ and $\ttheta\in\TTheta$ such that $\Cg_{\nq}(\ttheta,\xx)$ outputs a function $f_{\ttheta}$ with 
\begin{equation}\label{eq:Linftyuniversality}
\sup_{\xx \in \mathcal{X}} |f(\xx)-f_{\ttheta}(\xx)|  \leq \varepsilon.
\end{equation}
\end{corollary}
{The proof of Corollary~\ref{cor:3} is given in Section~\ref{subsec:proof6}. Corollary~\ref{cor:3} proves that the quantum circuit from Section~\ref{subsec:circuit} is universal also in a uniform sense: any continuous function can be approximated arbitrarily well, uniformly on compact sets using our variational quantum circuits. }

\subsection{$L^\infty$-error bounds for reservoir quantum circuits}

{Finally, we turn to the situation of quantum random neural network circuits as considered in Section~\ref{sec:ReservoirQ}.}
For this class of circuits we{,} can also obtain a uniform approximation result for {functions in } a subset of  $\Ffov_b \subset \Ffov$. 
Let~$\Ffov_b$ be the set of functions in~$\Ffov$ for which $\widehat{f}/\pi_a$ is bounded.
In particular, given an arbitrary compact set, the result below guarantees existence of readout weights for the reservoir quantum circuit which achieve an arbitrarily small uniform approximation error.
Suppose that
\[
\EE\left[\|\mathbf{A}_1\|^2\right]^{\frac{1}{2}}
= \left(\int_{\RR^d} \|\bx\|^2 \pi_a(\bx) \D \bx \right)^{\frac{1}{2}} < \infty.
\]
	 
\begin{theorem}\label{thm:ApproxUniformStronger} 
	For any $M>0$, $f \in \Ffov_b$ and
	$n\in\NN$, there exists an 
 $\RR^n$-valued $\sigma(\Af,\Bf)$-measurable random vector~$\Wf$ such that 
$$
\begin{aligned}
& \EE\left[\sup_{\xx \in [-M,M]^d}  \left|F_{\Wf}(\xx)-f(\xx)\right|\right] 
\\ & \leq \frac{1}{\sqrt{n}} \left(8 \left\|\frac{\widehat{f}}{\pi_a}\right\|_\infty  \left(\frac{\pi}{2^{3/2}}
+ 2 \pi M d^{\frac{1}{2}}
\EE[\|\mathbf{A}_1\|^2]^{\frac{1}{2}}
\right)  +  \Lbarf \right).
\end{aligned}
$$
\end{theorem}

{The proof of Theorem~\ref{thm:ApproxUniformStronger} is given in Section~\ref{subsec:proof7}.}

\subsection{{$L^\infty$-universal approximation by random quantum neural networks}}

{From Theorem~\ref{thm:ApproxUniformStronger} we also obtain a uniform universal approximation result for random quantum neural networks:
\begin{corollary} \label{cor:universalityRandomUniform}
Let $M>0$, $f \in \Cc([-M,M]^d,\RR)$ and assume that $\pi_a$ is continuous, strictly positive, decays at most polynomially and has finite second moments.
Then for any $\varepsilon >0$ there exist $n \in \N$ and an $\RR^n$-valued $\sigma(\Af,\Bf)$-measurable random vector~$\Wf$
such that  
\begin{equation}\label{eq:universalityRandomUniform}
\EE\left[\sup_{\xx \in [-M,M]^d} \left|F_{\Wf}(\xx)-f(\xx)\right|\right] \leq  \varepsilon.
\end{equation}
\end{corollary}}
{The proof of this result follows analogously to  Corollary~\ref{cor:3}, using the fact that for a Schwartz function $g$ the function $\widehat{g}/\pi_a$ is bounded when $\pi_a$ decays at most polynomially and arguing that $\overline{L}_{2}[g]$ is finite under the hypotheses on~$\pi_a$.
Corollary~\ref{cor:universalityRandomUniform} shows that the constructed random quantum circuit is universal also uniformly: any continuous function can be approximated arbitrarily well, uniformly on compact sets using the considered random quantum circuits.
}

{\section{Conclusion}\label{sec:conlcusion}
We have constructed universal variational quantum circuits and proved bounds on the approximation error. A key feature in all the derived bounds is that the approximation error rates do not deteriorate as the dimension of the input data increases. 
For both fully trainable and randomized quantum circuits{} we have considered several classes of functions{} and derived error bounds with respect to both mean-square and uniform metrics. }

{
Firstly, we have considered the large function class of integrable, continuous functions with integrable Fourier transforms. We have constructed trainable variational quantum circuits that are able to approximate such functions in a mean-square sense{,} up to explicitly given error. More generally, we have also obtained a universal approximation result for the same circuit: any square-integrable function can be approximated arbitrarily well by such a quantum neural network in the mean-square sense. }

{Secondly, we have considered a variant of the above quantum circuit{,} in which the parameters inside the circuit are randomly generated and only a final layer of weights after measuring the quantum state is trainable. We have proved that functions with sufficiently integrable  Fourier transform can be approximated in the mean-square sense by such quantum random neural networks. In addition to these quantitative bounds{,} we have also obtained a universal approximation result for quantum random neural networks and quantum extreme learning machines. 
Randomizing the parameters inside the circuit is a key feature of these systems in supervised machine learning applications. It brings the advantage that training is less complex, as it reduces to a linear regression.
}

{Finally, we have extended all these results to the situation when the error is measured uniformly on a compact set, rather than in a mean-square sense.
Depending on the situation and on the nature of inputs (deterministic or stochastic) either of the two error criteria are more suitable. 
}
\section{Proofs}\label{sec:proofs}

\subsection{Proof of Theorem~\ref{thm:Approx}}\label{subsec:proof1}
In this section we prove two propositions, the combination of which directly implies Theorem~\ref{thm:Approx}. {For notational convenience we write $l^{i}(\xx) := b^{i}+\ab^{i} \cdot \xx$. }

\begin{proposition} \label{prop:circuitOutput}
For any $n\in\NN$, $\ttheta\in\TTheta$, 
the representation
$\ftrn = \gtrn$ holds over~$\RR^d$.
\end{proposition}
\begin{proof}
Let us first calculate $\P_0^{n}$. To do this, we first compute 
$$
\begin{aligned}
\Ug \Vg \ket{0}^{\otimes \nq}
& = \Ug \ket{\psi}  =  \frac{1}{\sqrt{n}} \sum_{j=0}^{n-1} \Ug \ket{4j}  
\\ &  =
 \frac{1}{\sqrt{n}} \sum_{j=0}^{n-1} \sum_{k=0}^3\left[\Ug^{(j+1)}_1 \otimes \Ug^{(j+1)}_2\right]_{k+1,1} \ket{4j+k},
\end{aligned}
$$
and consequently
\begin{align*}
\P_0^{n}  & = \sum_{i=0}^{n-1} \left|\bra{4i} \Ug \Vg \ket{0}^{\otimes \nq} \right|^2
\\
 & =
 \sum_{i=0}^{n-1} \left|\bra{4i} \frac{1}{\sqrt{n}} \sum_{j=0}^{n-1} \sum_{k=0}^3\left[\Ug^{(j+1)}_1 \otimes \Ug^{(j+1)}_2\right]_{k+1,1} \ket{4j+k}  \right|^2
  \\ & =
 \frac{1}{n}\sum_{i=0}^{n-1} \left| \left[\Ug^{(i+1)}_1 \otimes \Ug^{(i+1)}_2\right]_{1,1}  \right|^2.
\end{align*}
Computing
\begin{align*}
[\Ug^{(i)}_1 \otimes \Ug^{(i)}_2]_{1,1}
& = [\Ug^{(i)}_1]_{1,1} [\Ug^{(i)}_2]_{1,1} 
\\ & = \cos\left(\frac{\gamma^{i}}{2}\right) 
\left[\Hg \begin{pmatrix} \E^{\frac{\I}{2} l^{i}(\xx)}
& 0 \\
0 & \E^{-\frac{\I}{2} l^{i}(\xx)} \end{pmatrix}  \Hg \right]_{1,1}\\
& =\frac{\cos\left(\frac{\gamma^{i}}{2}\right)}{\sqrt{2}} \left[\Hg \begin{pmatrix} \E^{\frac{\I}{2} l^{i}(\xx)}
& \E^{\frac{\I}{2} l^{i}(\xx)} \\ \E^{-\frac{\I}{2} l^{i}(\xx)}
& -\E^{-\frac{\I}{2} l^{i}(\xx)} \end{pmatrix}   \right]_{1,1}\\
& = \frac{1}{2}\cos\left(\frac{\gamma^{i}}{2}\right) \left(\E^{\frac{\I}{2} l^{i}(\xx)} + \E^{-\frac{\I}{2} l^{i}(\xx)}\right)\\
& = \cos\left(\frac{\gamma^{i}}{2}\right)
\cos\left(\frac{l^{i}(\xx)}{2}\right),
\end{align*} 
then
$
\P_0^{n} =
\frac{1}{n}\sum_{i=1}^{n} \cos\left(\frac{\gamma^{i}}{2}\right)^2
\cos\left(\frac{l^{i}(\xx)}{2}\right)^2$, 
which simplifies, using $\cos(y)^2=\frac{\cos(2y)+1}{2}$, to
\begin{align*}
\P_0^{n} & =
\frac{1}{n}\sum_{i=1}^{n} 
\frac{1}{4}\left(\cos\left(\gamma^{i}\right)+1\right)
\left(\cos\left(l^{i}(\xx)\right)+1\right)\\
& = \frac{1}{4} + \frac{1}{4n}\sum_{i=1}^{n} \cos\left(\gamma^{i}\right)
\cos\left( l^{i}(\xx)\right)
+ \frac{1}{4n}\sum_{i=1}^{n} \cos\left(\gamma^{i}\right)
\\ & \quad + \frac{1}{4n}
\sum_{i=1}^{n} \cos\left(l^{i}(\xx)\right).
\end{align*}

Next, for $m \in \{1,2,3\}$, we have
{\small 
\begin{align*}
& \P_m^{n}  = \sum_{i=0}^{n-1} \left|\bra{4i+m} \Ug \Vg \ket{0}^{\otimes \nq} \right|^2
\\ & =
\sum_{i=0}^{n-1} \left|\bra{4i+m} \frac{1}{\sqrt{n}} \sum_{j=0}^{n-1} \sum_{k=0}^{3}
\left[\Ug^{(j+1)}_1 \otimes \Ug^{(j+1)}_2\right]_{k+1,1} \ket{4j+k} \right|^2
\\ & =
\frac{1}{n}\sum_{i=0}^{n-1} 
\left| \left[\Ug^{(i+1)}_1 \otimes \Ug^{(i+1)}_2\right]_{m+1,1}  \right|^2.
\end{align*}
}
Computing as above 
\begin{align*}
[\Ug^{(i)}_1 \otimes \Ug^{(i)}_2]_{2,1} & = [\Ug^{(i)}_1]_{1,1} [\Ug^{(i)}_2]_{2,1} = \sin\left(\frac{\gamma^{i}}{2}\right)
\cos\left(\frac{l^{i}(\xx)}{2}\right),\\
[\Ug^{(i)}_1 \otimes \Ug^{(i)}_2]_{3,1} & = [\Ug^{(i)}_1]_{2,1} [\Ug^{(i)}_2]_{1,1} = 
\I \cos\left(\frac{\gamma^{i}}{2}\right)
\sin\left(\frac{l^{i}(\xx)}{2}\right),\\
[\Ug^{(i)}_1 \otimes \Ug^{(i)}_2]_{4,1} & = [\Ug^{(i)}_1]_{2,1} [\Ug^{(i)}_2]_{2,1} = 
\I \sin\left(\frac{\gamma^{i}}{2}\right)
\sin\left(\frac{l^{i}(\xx)}{2}\right),
\end{align*}
thus yields
\begin{align*}
\P_1^{n} & = \frac{1}{n}\sum_{i=1}^{n} \sin\left(\frac{\gamma^{i}}{2}\right)^2
\cos\left(\frac{l^{i}(\xx)}{2}\right)^2
\\
\P_2^{n} & = \frac{1}{n}\sum_{i=1}^{n} \cos\left(\frac{\gamma^{i}}{2}\right)^2
\sin\left(\frac{l^{i}(\xx)}{2}\right)^2
\\ 
\P_3^{n} & = \frac{1}{n}\sum_{i=1}^{n} \sin\left(\frac{\gamma^{i}}{2}\right)^2
\sin\left(\frac{l^{i}(\xx)}{2}\right)^2.
\end{align*}
Therefore,
\begin{align*}
\P_0^{n} + \P_1^{n} & = \frac{1}{n}\sum_{i=1}^{n}
\cos\left(\frac{l^{i}(\xx)}{2}\right)^2
= \frac{1}{2} + \frac{1}{2n}\sum_{i=1}^{n} \cos\left(l^{i}(\xx)\right),\\
\P_0^{n} + \P_2^{n} & = \frac{1}{n}\sum_{i=1}^{n} \cos\left(\frac{\gamma^{i}}{2}\right)^2 = \frac{1}{2} + \frac{1}{2n}\sum_{i=1}^{n} \cos\left(\gamma^{i}\right).
\end{align*}
Putting it all together we obtain, for any given $R>0$, that
\[
\begin{aligned}
& R - 2R\left[\P_1^{n}+\P_2^{n}\right]
\\ \quad  & = R\left[1 + 4\P_0^{n} - 2\left(\P_0^{n} + \P_1^{n}\right) - 2\left(\P_0^{n} + \P_2^{n}\right)\right]
\\ & = \frac{1}{n}\sum_{i=1}^{n} R\cos\left(\gamma^{i}\right) \cos\left(l^{i}(\xx)\right).
\end{aligned}
\]
\end{proof}

\begin{remark}
{
In the case of the alternative circuit in Remark~\ref{rmk:alternative} the proof follows by the same argument. Denote by $\Ug\Vg$ the circuit, and recall that it maps $\ket{0}^{\otimes 2n}$ 
to 
$\frac{1}{\sqrt{n}} \sum_{j=0}^{n-1} (\ket{1} \otimes \ket{1} )^{\otimes j} (\Ug_1^{(j+1)}\ket{0} \otimes \Ug_2^{(j+1)} \ket{0} ) \otimes (\ket{0} \otimes \ket{0} )^{\otimes n-j-1}$. Denote for $ i \in \{0,\ldots,n-1\}$, $m\in\{0,1,2,3\}$ the states
$\ket{\psi_{i,m}} =  (\ket{1} \otimes \ket{1} )^{\otimes i} \otimes (\ket{x} \otimes \ket{y} ) \otimes (\ket{0} \otimes \ket{0} )^{\otimes n-i-1}$ with $x,y \in \{0,1\}$  the coefficients in the binary representation $m=2x+y$.
Then, for example, $\P_0^{n}$ reads
\begin{align*}
	&\P_0^{n}   = \sum_{i=0}^{n-1} \left| \bra{\psi_{i,0}} \Ug_{i} \Vg \ket{0}^{\otimes 2n} \right|^2
	\\
	& =
	\sum_{i=0}^{n-1} \left|\bra{\psi_{i,0}}  \frac{1}{\sqrt{n}} \sum_{j=0}^{n-1} \sum_{m=0}^3\left[\Ug^{(j+1)}_1 \otimes \Ug^{(j+1)}_2\right]_{m+1,1} \ket{\psi_{j,m}}  \right|^2
	\\ & =
	\frac{1}{n}\sum_{i=0}^{n-1} \left| \left[\Ug^{(i+1)}_1 \otimes \Ug^{(i+1)}_2\right]_{1,1}  \right|^2, 
\end{align*} 
which is identical to the analogous expression obtained in the proof of Proposition~\ref{prop:circuitOutput}. 
}
\end{remark}

\begin{proposition} \label{prop:FourierApprox} 
 Let $R>0$ and $f\in\Ff_{R}$.
 For any $n \in \N$, there exists $\ttheta\in\TTheta$ such that
$$
\left(\int_{\RR^d} \left|\ftrn(\xx) - f(\xx)\right|^2 \mu(\D \xx)\right)^{1/2} \leq \frac{\LOnefhat}{\sqrt{n}}.
$$
\end{proposition}

\begin{proof}
Since $\widehat{f}\in L^{1}(\RR^d)$, the Fourier inversion theorem states that for all $\xx\in \RR^d$,
\[
f(\xx) = \int_{\RR^d} \E^{2\pi \I\xx\cdot \bx}\widehat{f}(\bx) \D \bx.
\]
Since~$f$ is real-valued, we may then write, for any $\xx\in\RR^d$,
\begin{align}\label{eq:auxEq2}
& f(\xx) = \int_{\RR^d} \E^{2\pi \I\xx\cdot \bx}\widehat{f}(\bx) \D \bx
\\ & = \int_{\RR^d}\left\{ \cos{(2\pi \xx \cdot \bx)}\mathrm{Re}[\widehat{f}(\bx)] - \sin{(2\pi \xx \cdot \bx)}\mathrm{Im}[\widehat{f}(\bx)] \right\}\D \bx 
\\
& = 
\int_{\RR^d}\left\{ \cos{(2\pi \xx \cdot \bx)}\mathrm{Re}[\widehat{f}(\bx)] \right. \\ & \left. \qquad \qquad + \cos\left(2\pi \xx \cdot \bx+\frac{\pi}{2}\right)\mathrm{Im}[\widehat{f}(\bx)] \right\}\D \bx.\nonumber
\end{align}
Let $\pf := \LOnefhat^{-1}\int_{\RR^d} |\mathrm{Re}[\widehat{f}(\bx)]| \D \bx \in [0,1]$,
$Z_1,\ldots,Z_n$ i.i.d.\ $\pf$-Bernoulli random variables so that
$\P(Z_i=1) = \pf$ and $\P(Z_i=0) = 1-\pf = \LOnefhat^{-1}\int_{\RR^d} |\mathrm{Im}[\widehat{f}(\bx)]| \D\bx$. 
If $\int_{\RR^d} |\mathrm{Re}[\widehat{f}(\bx)]| \D\bx \neq 0 $, let $\nu_1$ be the probability measure on $\RR^d$ with density $\frac{|\mathrm{Re}[\widehat{f}]|}{\int_{\RR^d} |\mathrm{Re}[\widehat{f}(\bx)]| \D\bx }$, otherwise $\nu_1$ is an arbitrary probability measure on $\RR^d$.
Analogously, if $\int_{\RR^d} |\mathrm{Im}[\widehat{f}(\bx)]| \D\bx \neq 0$, let $\nu_0$ be the probability measure on $\RR^d$ with density $\frac{|\mathrm{Im}[\widehat{f}]|}{\int_{\RR^d} |\mathrm{Im}[\widehat{f}(\bx)]| \D\bx }$, otherwise $\nu_0$ is an arbitrary probability measure on $\RR^d$.
Let $\bU_1,\ldots,\bU_n$ (resp. $\bV_1,\ldots,\bV_n$) be i.i.d.\ random variables with distribution~$\nu_1$ (resp.~$\nu_0$) and assume that $\bU_1,\ldots,\bU_n, \bV_1,\ldots,\bV_n, Z_1,\ldots,Z_n$ are independent. 
Set
\begin{align*}
\mathbf{A}_i & := 2 \pi (Z_i \bU_i+(1-Z_i)\bV_i),
\qquad
B_i := \frac{\pi}{2}(1-Z_i),\\
\quad
W_i & := L^{1}[\widehat{f}]\left[ \frac{\mathrm{Re}[\widehat{f}](\bU_i)}{|\mathrm{Re}[\widehat{f}](\bU_i)|} Z_i + \frac{\mathrm{Im}[\widehat{f}](\bV_i)}{|\mathrm{Im}[\widehat{f}](\bV_i)|}(1-Z_i) \right],
\end{align*}
with the quotient set to zero when the denominator is null, and consider the random function 
$$
F(\xx) := \frac{1}{n}\sum_{i=1}^n W_i \cos(B_i+\mathbf{A}_i\cdot \xx).
$$
We calculate 
\[
\begin{aligned}
&\LOnefhat \pf\EE\left[\frac{\mathrm{Re}[\widehat{f}](\bU_1)}{|\mathrm{Re}[\widehat{f}](\bU_1)|}  \cos(2\pi \bU_1\cdot \xx)\right] 
\\ & =  \LOnefhat 
\pf\int_{\RR^d} \frac{\mathrm{Re}[\widehat{f}](\bx)}{|\mathrm{Re}[\widehat{f}](\bx)|}  \cos(2\pi \bx \cdot \xx) \nu_1(\D\bx) 
\\ & = \LOnefhat \pf\int_{\RR^d} \frac{\mathrm{Re}[\widehat{f}](\bx)}{\int_{\RR^d} |\mathrm{Re}[\widehat{f}(z)]| \D z}  \cos(2\pi \bx \cdot \xx) \D\bx  
\\ & = \int_{\RR^d} {\mathrm{Re}[\widehat{f}](\bx)}  \cos(2\pi \bx \cdot \xx) \D\bx,
\end{aligned}
\]
and similarly
\[
\begin{aligned}
 & \LOnefhat  (1-\pf)\EE\left[\frac{\mathrm{Im}[\widehat{f}](\bV_1)}{|\mathrm{Im}[\widehat{f}](\bV_1)|}
\cos\left(\frac{\pi}{2}+2 \pi \bV_1\cdot\xx\right)\right]
\\ & =   \int_{\RR^d} {\mathrm{Im}[\widehat{f}](\bx)} \cos\left(\frac{\pi}{2}+2 \pi \bx \cdot \xx\right) \D\bx 
\end{aligned}
\]
yielding
\[
\begin{aligned}
 & \EE[F(\xx)] = \EE[W_1 \cos(B_1+\mathbf{A}_1\cdot \xx)]
\\ & =  \LOnefhat \left\{\pf\EE\left[\frac{\mathrm{Re}[\widehat{f}](\bU_1)}{|\mathrm{Re}[\widehat{f}](\bU_1)|}  \cos(2\pi \bU_1\cdot \xx)\right]  \right. 
\\ & \left. \, + (1-\pf)\EE\left[\frac{\mathrm{Im}[\widehat{f}](\bV_1)}{|\mathrm{Im}[\widehat{f}](\bV_1)|}
\cos\left(\frac{\pi}{2}+2 \pi \bV_1\cdot\xx\right)\right]\right\}
 =  f(\xx). 
\end{aligned}
\]
In particular, using the i.i.d.\ assumption and Fubini, we obtain
\begin{equation}\label{eq:L2estimate}\begin{aligned}
& \EE\left[ \int_{\RR^d} |f(\xx) - F(\xx)|^2 \mu(\D \xx) \right]
 = \int_{\RR^d} \VV[F(\xx)] \mu(\D \xx) 
\\ & = \frac{1}{n^2} \int_{\RR^d} \VV\left[\sum_{i=1}^n W_i \cos(B_i + \mathbf{A}_i\cdot \xx)\right] \mu(\D\xx)
\\ & = \frac{1}{n} \int_{\RR^d} \VV\left[ W_1 \cos(B_1 + \mathbf{A}_1\cdot \xx)\right] \mu(\D\xx)
\\ & \leq \frac{1}{n} \int_{\RR^d} \EE\left[\left(W_1 \cos(B_1 + \mathbf{A}_1\cdot \xx)\right)^2\right] \mu(\D\xx)
\\ & \leq \frac{1}{n}  \EE\left[W_1^2\right] 
 \leq  \frac{1}{n} \left(\int_{\RR^d} |\widehat{f}(\bx)| \D\bx\right)^2
= \frac{1}{n} \LOnefhat^2.
\end{aligned}
\end{equation}
Since $\P(Z\leq B)>0$ for any non-negative random variable~$Z$ with $\EE[Z]\leq B$ for $B>0$, then
there exists $\omega \in \Omega$ such that 
$F_\omega(\xx) = \frac{1}{n}\sum_{i=1}^n W_i(\omega) \cos(B_i(\omega)+\mathbf{A}_i(\omega)\cdot\xx)$ satisfies
$$
\int_{\RR^d} |f(\xx)-F_\omega(\xx)|^2 \mu(\D\xx) \leq  \frac{1}{n} \left(\int_{\RR^d} |\widehat{f}(\bx)| \D\bx\right)^2.
$$
It remains to show that $F_\omega=f_{\ttheta}$ for a suitable choice of weights
$\ttheta=(\ab^{i},b^{i},\gamma^{i})_{i=1,\ldots,n}$.
This follows directly by choosing $b^{i}=B_i(\omega)$, 
$\ab^{i}=\mathbf{A}_i(\omega)$ and $\gamma^{i} = \arccos(\frac{W_i(\omega)}{R})$ (so that $R\cos\left(\gamma^{i}\right)= W_i(\omega) $), which is well defined because 
$W_i(\omega) = \LOnefhat \phi_i$ for some $\phi_i \in \{-1,1\}$ and thus $|\frac{W_i(\omega)}{R}|\leq 1$ given the constraint $\LOnefhat\leq R$. Therefore, 
$$
\left(\int_{\RR^d} |f(\xx)-f_{\ttheta}(\xx)|^2 \mu(\D\xx)\right)^{1/2} \leq  \frac{\LOnefhat}{\sqrt{n}}.
$$
\end{proof}

\subsection{Proof of Corollary~\ref{cor:universality}}\label{subsec:proof2}
\begin{proof}
We first show that~$f$ can be approximated up to error $\frac{\varepsilon}{2}$ by a function in 
$\Cc_c^\infty(\RR^d)$. This follows by standard arguments, which we give for completeness. 
By \cite[Lemma~1.35]{Kallenberg2002} there exists a bounded continuous function $g \colon \RR^d \to \R$ with $\left|\int_{\RR^d} |f(\xx)-g(\xx)|^2 \mu(\D\xx)\right|^{\half} \leq  \frac{\varepsilon}{6}$. Denote by $\chi_m \colon \RR^d \to [0,1]$ a continuous function with $\chi_m(\xx)=1$ for $\xx \in [-m,m]^d$ and $\chi_m(\xx)=0$ for $\xx \in \RR^d \setminus [-m-1,m+1]^d$. Then 
$\left(\int_{\RR^d} |g(\xx)-\chi_m(\xx)g(\xx)|^2 \mu(\D\xx)\right)^{\half} \leq \sup_{x \in \RR^d } |g(\xx)| \mu(\{\RR^d \setminus [-m,m]^d\}) \leq \frac{\varepsilon}{6}$ for $m$ chosen sufficiently large. 
Finally, since $\Cc_c^\infty(\RR^d)$ is dense in $\Cc_c(\RR^d)$ in the supremum norm, 
there exists $h \in \Cc_c^\infty(\RR^d)$ with 
\begin{align}\label{eq:auxEq1}
 & \left(\int_{\RR^d} |f(\xx)-h(\xx)|^2 \mu(\D\xx)\right)^{\half} \\
 & \leq \|f-g\|_{L^2(\mu)} + \|g-\chi_m g \|_{L^2(\mu)} + \|\chi_m g - h \|_{L^2(\mu)} \leq \frac{\varepsilon}{2}.\nonumber
\end{align}

Now we show how to apply Theorem~\ref{thm:Approx} to~$h$. 
Since~$h$ is a Schwartz function, its Fourier transform $\widehat{h}$ is as well and~$h$ 
and~$\widehat{h}$ are both integrable. With $n=\lceil (2C_h \varepsilon^{-1})^2 \rceil$ and $R = L^{1}[\widehat{h}]$, Theorem~\ref{thm:Approx} yields the existence of $\ttheta\in\TTheta$ such that 
\[
\left(\int_{\RR^d} |h(\xx)-f_{\ttheta}(\xx)|^2 \mu(\D\xx)\right)^{\half} \leq  \frac{ L^{1}[\widehat{h}]}{\sqrt{n}} \leq \frac{\varepsilon}{2}. 
\]
This estimate together with~\eqref{eq:auxEq1} then imply~\eqref{eq:universality}, as claimed. 
\end{proof}

\subsection{Proof of Theorem~\ref{thm:random}}\label{subsec:proof3}
Propositions~\ref{prop:FourierApproxRandom} and~\ref{prop:circuitOutputRandom} are the essential tools, directly implying Theorem~\ref{thm:random}.
For $\ww \in\RR^n$, mimicking~\eqref{eq:fct},
introduce the (random) map $G_{\ww}:\RR^n\to\RR$ and,
for each $i=1,\ldots, n$, the linear map $L^i$ by

\begin{equation}\label{eq:randomG}
\begin{aligned}
G_{\ww}(\xx) & := \frac{1}{n}\sum_{i=1}^{n} W_i \cos\left(L^i(\xx)\right),
\\
L^i(\xx)  & := \frac{\pi}{2} B^{i}+2 \pi \Af^{i} \cdot \xx.
\end{aligned}
\end{equation}

\begin{proposition}\label{prop:FourierApproxRandom}
 For any $f \in\Ffov$, there exist a $\sigma(\Af,\Bf)$-measurable random vector~$\Wf$ such that the random function~$G_{\Wf}$ in~\eqref{eq:randomG} satisfies
$$
\EE\left[\int_{\RR^d} \left|G_{\Wf}(\xx)-f(\xx)\right|^2 \mu(\D \xx)\right]^{1/2} \leq \frac{\Lbarf}{\sqrt{n}}.
$$
\end{proposition}
\begin{proof}
For each $i\in\{1,\ldots, n\}$, the random variable
\[
W_i := \frac{2}{\pi_a(\mathbf{A}^{i})} 
\left\{\left(1-B^{i}\right) \mathrm{Re}[\widehat{f}](\mathbf{A}^{i})
+ B^{i}\mathrm{Im}[\widehat{f}](\mathbf{A}^{i}) \right\},
\]
is well defined since $|\widehat{f}|\ll \pi_a$. 
The independence and the i.i.d.\ assumptions then yield
\[
\begin{aligned}
& \EE\left[G_{\Wf}(\xx)\right]
 = \EE\left[W_1 \cos\left(L^1(\xx)\right)\right]
\\ &  =  \int_{\RR^d} 
\left\{\mathrm{Re}[\widehat{f}](\bx) \cos(2 \pi \bx \cdot \xx)  \right.
\\ &  \left. \qquad  + \mathrm{Im}[\widehat{f}](\bx) \cos\left(\frac{\pi}{2}+2 \pi \bx \cdot \xx\right)\right\} \D\bx
  = f(\xx),
\end{aligned}
\]
where the last step follows by~\eqref{eq:auxEq2}. 
Proceeding as in the proof of Proposition~\ref{prop:FourierApprox}, we obtain
\[\begin{aligned}
& \EE\left[ \int_{\RR^d} \left|f(\xx) - G_{\Wf}(\xx)\right|^2 \mu(\D\xx) \right] = \int_{\RR^d} \VV\left[G_{\Wf}(\xx)\right] \mu(\D\xx) 
	\\ & = \frac{1}{n^2} \int_{\RR^d} \VV\left[\sum_{i=1}^n W_i \cos\left(L^i(\xx)\right)\right] \mu(\D\xx)
	\\ & = \frac{1}{n} \int_{\RR^d} \VV\left[ W_1 \cos\left(L^1(\xx)\right)\right] \mu(\D\xx)
	\\ & \leq \frac{1}{n} \int_{\RR^d} \EE\left[\left(W_1 \cos\left(L^1(\xx)\right)\right)^2\right] \mu(\D\xx)
	\\ & \leq \frac{1}{n}  \EE\left[W_1^2\right]
	  =  \frac{2}{n} \int_{\RR^d} 
	\frac{1}{\pi_a(\bx)}  \left( \mathrm{Re}[\widehat{f}](\bx)^2 + \mathrm{Im}[\widehat{f}](\bx)^2  \right) \D\bx,
\end{aligned}
\]
which implies the claimed bound.
\end{proof}

\begin{remark} As mentioned earlier, the proof above yields an analogous result for the circuit constructed in Section~\ref{sec:variational}. 
Indeed, choosing $R=\max(W_1,\ldots,W_n)$, 
then
$\Gamma^{i} = \arccos(\frac{W_i}{R})$ is well defined as a $\sigma(\Af,\Bf)$-measurable random variable, and so is~$R$, so that the random function
$\overline{F}(\xx) :=  \frac{1}{n}\sum_{i=1}^{n} R\cos\left(\Gamma^{i}\right) \cos(\frac{\pi}{2} B^{i}+2 \pi \Af^{i} \cdot \xx)$
satisfies
\[
\EE\left[\int_{\RR^d} \left|\overline{F}(\xx)-f(\xx)\right|^2 \mu(\D\xx)\right]^{1/2} \leq \frac{\Lbarf}{\sqrt{n}}.
\]
\end{remark}

The next result shows that the functions~$G_{\ww}$ in Proposition~\ref{prop:FourierApproxRandom} can be realised by the considered quantum circuit.

\begin{proposition}\label{prop:circuitOutputRandom}
For $n \in \N$, given~$\overline{\Cg}_{\nqb}(\ttheta)$, and $\ww \in \RR^n$,
the identity $F_{\ww}\equiv G_{\ww}$ holds on~$\RR^n$.
\end{proposition}
\begin{proof}
The proof requires the computation of the probabilities~$\PPb_k$ for~$k$ even. 
First,
$$
\begin{aligned}
& \overline{\Ug}\, \overline{\Vg} \ket{0}^{\otimes \nqb}  = \overline{\Ug} \ket{\overline{\psi}}  =  \frac{1}{\sqrt{n}} \sum_{j=0}^{n-1} \overline{\Ug} \ket{2j}  
\\ &  =
\frac{1}{\sqrt{n}} \sum_{j=0}^{n-1}  \left\{\left[\overline{\Ug}^{(j+1)}_1\right]_{1,1} \ket{2j} +  \left[\overline{\Ug}^{(j+1)}_1\right]_{2,1} \ket{2j+1}\right\}.
\end{aligned}
$$
Therefore, for $k=2j+m$, $j \in \{0,\ldots,n-1\}, m \in \{0,1\}$, 
\begin{equation}\label{eq:auxEq3} 
\PPb_k  =  \left|\bra{k} \overline{\Ug}\, \overline{\Vg} \ket{0}^{\otimes \nqb}\right|^2
 =
 \frac{1}{n} 
 \left| \left[\overline{\Ug}^{(j+1)}_1\right]_{m+1,1}  \right|^2.
\end{equation}
Similarly as in the proof of Proposition~\ref{prop:circuitOutput} we now obtain for $i=j+1$ that
\[ \begin{aligned}
&\left[\overline{\Ug}^{(i)}_1\right]_{1,1}  = \left[\Hg \begin{pmatrix} \E^{\frac{\I}{2} L^i(\xx)} & 0 \\ 0 & \E^{-\frac{\I}{2} L^i(\xx)} \end{pmatrix}  \Hg \right]_{1,1}
\\ & 	 =\frac{1}{\sqrt{2}} \left[\Hg \begin{pmatrix} \E^{\frac{\I}{2} L^i(\xx)} & \E^{\frac{\I}{2} L^i(\xx)} \\ \E^{-\frac{\I}{2} L^i(\xx)} & -\E^{-\frac{\I}{2} L^i(\xx)} \end{pmatrix}   \right]_{1,1}
	 = \cos\left(\frac{L^i(\xx)}{2}\right).
\end{aligned} 
\]
By plugging this into~\eqref{eq:auxEq3} and using the double-angle formula we then obtain
$$
\PPb_{2(i-1)}
=  \frac{1}{n} \left|\cos\left(\frac{1}{2}L^i(\xx)\right)\right|^2
=  \frac{1}{2n} \cos\left(L^i(\xx)\right) + \frac{1}{2n}.
$$
Inserting this expression into $F_{\ww}$ in~\eqref{eq:randomF} then yields the claimed representation. 
\end{proof}

\subsection{Proof of Corollary~\ref{cor:2}} \label{subsec:proof4}
\begin{proof}
	The assumption on $f$ guarantees that $f \in \Cc(\RR^d)$ and  $\widehat{f} \in L^1(\RR^d)$. In addition, $\pi_a$ is strictly positive hence $|\widehat{f}|\ll \pi_a$ because $\pi_a(\xx)\neq 0$ for all $\xx \in \RR^d$. 
	Finally, $\pi_a \geq \delta t_\nu$ and thus inserting $t_\nu(\bx) = \frac{\Gamma((\nu+d)/2)}{\Gamma(\nu/2) \nu^{d/2} \pi^{d/2}} (1 + \nu^{-1} \|\bx\|^2)^{-\frac{\nu}{2}}$ yields, with $c:=\frac{1}{\delta} \frac{\Gamma(\frac{\nu}{2}) \nu^{d/2} \pi^{d/2}}{\Gamma(\frac{\nu+d}{2})}$,
	\begin{equation} \label{eq:auxEq4} \begin{aligned}
			& \int_{\RR^d} \frac{|\widehat{f}(\bx)|^2}{\pi_a(\bx)} \D\bx  \leq 
			\int_{\RR^d} \frac{|\widehat{f}(\bx)|^2}{\delta t_\nu(\bx)} \D\bx 
			\\ & \leq 
			c	\int_{\RR^d} |\widehat{f}(\bx)|^2(1 + \nu^{-1} \|\bx\|^2)^{\frac{\nu+d}{2}} \D\bx 
			\\ & \leq 
	c	 \max\left(1,\frac{1}{\nu}\right)^{\frac{\nu+d}{2}}	\int_{\RR^d} |\widehat{f}(\bx)|^2(1 + \|\bx\|^2)^{\frac{\nu+d}{2}} \D\bx,
	\end{aligned}	
 \end{equation}
	which is finite since $f \in \Hh^s(\RR^d)$ and $s =  (\nu+d)/2$. 
	The hypotheses of Theorem~\ref{thm:random} are thus satisfied and so the claim follows from Theorem~\ref{thm:random} and the bound~\eqref{eq:auxEq4}. 
\end{proof}

\subsection{Proof of Theorem \ref{thm:ApproxUniform}}\label{subsec:proof5}
\begin{proof} The proof is identical to that of Theorem~\ref{thm:Approx}, except that the $L^2$-bound in \eqref{eq:L2estimate} is replaced by an $L^\infty$-bound, and we use the same notation as there. 
Define first for all $\xx \in [-M,M]^d$, $i=1,\ldots,n$ the random variables  $U_{i,\xx} = W_i \cos(B_i+\mathbf{A}_i\cdot \xx)$. Then the $L^\infty$-error that we aim to bound can be rewritten as 
\[\begin{aligned}
&	\EE\left[ \sup_{x \in [-M,M]^d} |f(\xx) - F(\xx)| \right]
	\\ &  = \EE\left[ \sup_{x \in [-M,M]^d} \left| 
 \frac{1}{n}\sum_{i=1}^n \left(U_{i,\xx}-\EE[U_{i,\xx}]\right)\right| \right] .
\end{aligned}
\]
Let $\varepsilon_1,\ldots,\varepsilon_d$ be i.i.d.\ Rademacher random variables independent of $\mathbf{A}=(\mathbf{A}_1,\ldots,\mathbf{A}_n)$ and $\mathbf{B}=(B_1,\ldots,B_n)$. Using symmetrisation, we then obtain 
\[\begin{aligned}
&	\EE\left[ \sup_{x \in [-M,M]^d} |f(\xx) - F(\xx)| \right]
\\	& \leq 2 \EE\left[ \sup_{x \in [-M,M]^d} \left| \frac{1}{n}\sum_{i=1}^n \varepsilon_i U_{i,\xx}\right| \right] .
\end{aligned}
\]
For any $\ab=(\ab_1,\ldots,\ab_n) \in (\R^d)^n$, $\mathbf{b}=(b_1,\ldots,b_n) \in \R^n$, 
consider the set 
\[T_{\ab,\mathbf{b}} 
:= \{
(l^{i}(\xx))_{i=1,\ldots,n} \,\colon \, \xx \in [-M,M]^d 
\},
\]
with the function~$l^i(\xx) = b_i+\ab_i\cdot \xx$,
and for $\mathbf{w}=(w_1,\ldots,w_n) \in \R^n$ with $|w_i|\leq L^{1}[\widehat{f}]$  define the maps $\varrho_{w_i} \colon \R \to \R$ via $\varrho_{w_i}(x) = \frac{w_i}{L^{1}[\widehat{f}]} (\cos(x)-1)$. 
Then independence yields that
{\small
\[\begin{aligned}
&\EE\left[ \sup_{\xx \in [-M,M]^d} \left| \frac{1}{n}\sum_{i=1}^n \varepsilon_i U_{i,\xx}\right| \right] \\ & = \EE\left[ \left. \EE\left[ \sup_{\xx \in [-M,M]^d} \left|  \frac{1}{n}\sum_{i=1}^n \varepsilon_i w_i \cos(l^{i}(\xx))  \right| \right]  \right|_{(\mathbf{w},\ab,\mathbf{b})=(\mathbf{W},\mathbf{A},\mathbf{B})} \right].
\end{aligned}\]
}
and with the above definitions and Jensen's inequality,
\[\begin{aligned}
& \EE\left[ \sup_{\xx \in [-M,M]^d} \left|  \frac{1}{n}\sum_{i=1}^n \varepsilon_i w_i \cos(l^{i}(\xx))  \right| \right] 
\\ & \leq  L^{1}[\widehat{f}]  \EE\left[ \sup_{\mathbf{t} \in T_{\ab,b}} \left|  \frac{1}{n}\sum_{i=1}^n \varepsilon_i \varrho_{w_i} (t_i) \right| \right] + \VV\left[\frac{1}{n}\sum_{i=1}^n \varepsilon_i w_i \right]^{\frac{1}{2}}.
\end{aligned}
\]
The fact that $\varrho_{w_i}(0)=0$ and that $\varrho_{w_i}$ is $1$-Lipschitz allows us to apply the comparison theorem \cite[Theorem~4.12]{Ledoux2013}
\[
\EE\left[ \sup_{\mathbf{t} \in T_{\ab,b}} \left|  \frac{1}{n}\sum_{i=1}^n \varepsilon_i \varrho_{w_i} (t_i) \right| \right] \leq 2\EE\left[ \sup_{\mathbf{t} \in T_{\ab,b}} \left|  \frac{1}{n}\sum_{i=1}^n \varepsilon_i t_i \right| \right].
\] 
Using this, the fact that $\varepsilon_1,\ldots,\varepsilon_n$ are i.i.d.\ Rademacher random variables we thus obtain
\[\begin{aligned}
	& \EE\left[ \sup_{\xx \in [-M,M]^d} \left|  \frac{1}{n}\sum_{i=1}^n \varepsilon_i w_i \cos(b_i+\ab_i\cdot \xx)  \right| \right]  
	\\ & \leq 2 L^{1}[\widehat{f}]  \EE\left[ \sup_{\mathbf{t} \in T_{\ab,b}} \left|  \frac{1}{n}\sum_{i=1}^n \varepsilon_i t_i \right| \right] + \VV\left[\frac{1}{n}\sum_{i=1}^n \varepsilon_i w_i \right]^{\frac{1}{2}}
	\\ & = 2 L^{1}[\widehat{f}]  \EE\left[ \sup_{\xx } \left|  \frac{1}{n}\sum_{i=1}^n \varepsilon_i (l^{i}(\xx)) \right| \right] 
 + \frac{1}{n} \left|\sum_{i=1}^n\VV\left[ \varepsilon_i w_i \right]\right|^{\frac{1}{2}}
	\\ & \leq 2 L^{1}[\widehat{f}]  \left( \EE\left[\left|  \frac{1}{n}\sum_{i=1}^n \varepsilon_i b_i \right| \right] \right. \\ & \qquad \qquad  \left.
	+
\EE\left[ \sup_{\xx \in [-M,M]^d} \left| \xx \cdot \frac{1}{n}\sum_{i=1}^n \varepsilon_i \ab_i  \right| \right]
	\right)  + \frac{\left\| \mathbf{w} \right\|}{n}.
\end{aligned}
\]
We already estimated the first term (with~$b_i$ instead of~$w_i$) as
\[
\EE\left[\left|  \frac{1}{n}\sum_{i=1}^n \varepsilon_i b_i \right| \right] \leq  \frac{\left\| \mathbf{b} \right\|}{n}.
\]
For the second term, since $\EE[\varepsilon_i \varepsilon_j]=\delta_{i,j}$, then
\[ 
\begin{aligned}
& \EE\left[ \sup_{\xx \in [-M,M]^d} \left| \xx \cdot \frac{1}{n}\sum_{i=1}^n \varepsilon_i \ab_i  \right| \right] \\ & \leq \sup_{\xx \in [-M,M]^d} \| \xx \|_{\infty}
\EE\left[  \left\|\frac{1}{n}\sum_{i=1}^n \varepsilon_i \ab_i  \right\|_1 \right] 
\\ & 
=
\frac{M}{n}
 \sum_{j=1}^d \EE\left[ \left|\sum_{i=1}^n  \varepsilon_i a_{i,j} \right| \right]
 \leq \frac{M}{n}
\sum_{j=1}^d \left(\sum_{i=1}^n a_{i,j}^2 \right)^{\frac{1}{2}}
\end{aligned}
\] 
and thus 
\[\begin{aligned}
	& \EE\left[ \sup_{\xx \in [-M,M]^d} \left|  \frac{1}{n}\sum_{i=1}^n \varepsilon_i w_i \cos(l^{i}(\xx))  \right| \right]  
\\ &  \leq 2 L^{1}[\widehat{f}]  \left( \frac{\left\| \mathbf{b} \right\|}{n}
	+
\frac{M}{n}
\sum_{j=1}^d \left(\sum_{i=1}^n a_{i,j}^2 \right)^{\frac{1}{2}}
	\right)  + \frac{\left\| \mathbf{w} \right\|}{n}.
\end{aligned}
\]
Combining these estimates with Jensen's inequality yields
{\small 
\[\begin{aligned}
	&\EE \left[ \sup_{x \in [-M,M]^d} \left| \frac{1}{n}\sum_{i=1}^n \varepsilon_i U_{i,\xx}\right| \right]  \\ &  = \EE\left[ \left. \EE\left[ \sup_{\xx \in [-M,M]^d} \left|  \frac{1}{n}\sum_{i=1}^n \varepsilon_i w_i \cos(l^{i}(\xx))  \right| \right]  \right|_{(\mathbf{w},\ab,\mathbf{b})=(\mathbf{W},\mathbf{A},\mathbf{B})} \right]
	\\  & \leq  2 L^{1}[\widehat{f}]  \left[\frac{\EE\left[ \left\| \mathbf{B} \right\|\right]}{n}
	+
	\frac{M }{n}
\sum_{j=1}^d \EE\left[ \left(\sum_{i=1}^n| A_{i,j}|^2 \right)^{\frac{1}{2}}\right]
	\right] + \frac{\EE\left[\left\| \mathbf{W} \right\|\right]}{n} 
	\\ & \leq 2 L^{1}[\widehat{f}]  \left(\frac{1}{\sqrt{n}} \EE[B_1^2]^{\frac{1}{2}}
	+ \frac{M}{\sqrt{n}}
\sum_{j=1}^d \EE\left[| A_{1,j}|^2 \right]^{\frac{1}{2}}
	\right)  + \frac{ \EE\left[ \left| W_1 \right|^2 \right]^{\frac{1}{2}}}{\sqrt{n}} 
	\\ & \leq 2 L^{1}[\widehat{f}]  \left({\frac{\pi}{2}} \frac{1}{\sqrt{n}} 
	+ \frac{M d^{\frac{1}{2}}}{\sqrt{n}}
	\EE[\|\mathbf{A}_1\|^2]^{\frac{1}{2}}
	\right)  + \frac{L^{1}[\widehat{f}]}{\sqrt{n}}.
\end{aligned}
\]
}
Inserting the definition of $\mathbf{A}_1$ yields
\[\begin{aligned}
&	\EE[\|\mathbf{A}_1\|^2]^{\frac{1}{2}}  = 	2 \pi \EE[\| Z_1 \bU_1+(1-Z_1)\bV_1\|^2]^{\frac{1}{2}}
\\ & = 	2 \pi \left(\pf \int_{\RR^d} \| \bx\|^2 \nu_1(\D \bx)  +(1-\pf)\int_{\RR^d} \| \bx\|^2 \nu_0(\D \bx) \right)^{\frac{1}{2}}
\\ & = 	2 \pi \LOnefhat^{-1} \left(  \int_{\RR^d} \| \bx\|^2 \left\{|\mathrm{Re}[\widehat{f}](\bx)| + |\mathrm{Im}[\widehat{f}](\bx)| \right\} \D \bx \right)^{\frac{1}{2}}
\\ & = 	2 \pi \LOnefhat^{-1/2 }\left( \int_{\RR^d} \| \bx\|^2 |\widehat{f}(\bx)|\D \bx \right)^{\frac{1}{2}}.
\end{aligned} 
\]
Overall, we obtain the bound
\[\begin{aligned}
&\EE\left[ \sup_{x \in [-M,M]^d} |f(\xx) - F(\xx)| \right]
\\ & \leq 4 L^{1}[\widehat{f}]  \left({\frac{\pi}{2}} \frac{1}{\sqrt{n}} 
+ \frac{M d^{\frac{1}{2}}}{\sqrt{n}}
\EE[\|\mathbf{A}_1\|^2]^{\frac{1}{2}}
\right)  + \frac{2 L^{1}[\widehat{f}]}{\sqrt{n}}
\\ & \leq \frac{2 (\pi+1) L^{1}[\widehat{f}] 
	+ 8 \pi  M d^{\frac{1}{2}}
	\LOnefhat^{\frac{1}{2}}\|f\|_{\mathcal{B}_2}}{\sqrt{n}}.
\end{aligned}
\]

\end{proof}

\subsection{Proof of Corollary~\ref{cor:3}} \label{subsec:proof6}
\begin{proof} 
The proof proceeds similarly as in Corollary~\ref{cor:universality}. 
We first show that~$f$ can be approximated on $\mathcal{X}$ up to error $\frac{\varepsilon}{2}$ by a function in 	$\Cc_c^\infty(\RR^d)$.  Choose $M>0$ such that $\mathcal{X} \subset [-M,M]^d$. 
By Tietze's extension theorem~\cite[Theorem~20.4]{Rudin1987} there exists a bounded continuous function $F\colon \R^d \to \R$ which coincides with~$f$ on $\mathcal{X}$. Now let $\chi \colon \RR^d \to [0,1]$ a continuous function with $\chi(\xx)=1$ for $\xx \in [-M,M]^d$ and $\chi(\xx)=0$ for $\xx \in \RR^d \setminus [-M-1,M+1]^d$. Then, since $\Cc_c^\infty(\RR^d)$ is dense in $\Cc_c(\RR^d)$ in the supremum norm, there exists $h \in \Cc_c^\infty(\RR^d)$ with 
\begin{equation}\label{eq:auxEqCor}
	\sup_{x \in \RR^d} |F(\xx)\chi(\xx)-h(\xx)|  \leq \frac{\varepsilon}{2}.
\end{equation}
We now show how to apply Theorem~\ref{thm:ApproxUniform} to~$h$. 
Since~$h$ is a Schwartz function, its Fourier transform $\widehat{h}$ is as well. In particular~$h$ 
and~$\widehat{h}$ are both integrable and so is $\bx \mapsto \|\bx\|^2 |\widehat{h}(\bx)|$ over~$\RR^d$. 
Choosing $R = L^{1}[\widehat{h}]$ and $n=\lceil (2(2 (\pi+1) L^{1}[\widehat{h}] 
	+ 8 \pi  M d^{\frac{1}{2}}
	L^{1}[\widehat{h}] ^{\frac{1}{2}}\|h\|_{\mathcal{B}_2}) \varepsilon^{-1})^2 \rceil $, Theorem~\ref{thm:ApproxUniform} yields the existence of $\ttheta\in\TTheta$ such that 
\[
\sup_{\xx \in [-M,M]^d} |h(\xx)-f_{\ttheta}(\xx)| \leq \frac{\varepsilon}{2}. 
\]
This estimate together with~\eqref{eq:auxEqCor} then imply
\begin{equation*}
\begin{aligned}
& 	\sup_{\xx \in \mathcal{X}} |f(\xx)-f_{\ttheta}(\xx)|
\\ & \leq 	\sup_{\xx \in \mathcal{X}} |F(\xx)\chi(\xx)-h(\xx)| + \sup_{\xx \in \mathcal{X}} |h(\xx)-f_{\ttheta}(\xx)|  \leq  \varepsilon.
\end{aligned}
\end{equation*}
\end{proof}

\subsection{Proof of Theorem \ref{thm:ApproxUniformStronger}}\label{subsec:proof7}
\begin{proof}
The proof follows by combining the arguments from the proofs of Theorem~\ref{thm:random} and Theorem~\ref{thm:ApproxUniform}. Analogously to the proof of the latter, we need to derive an $L^\infty$-error bound instead of the $L^2$-error bound in Proposition~\ref{prop:FourierApproxRandom}.
We may use $\EE[G_{\Wf}(\xx)]=f(\xx)$ to proceed precisely as in the proof of   Theorem~\ref{thm:ApproxUniform} and obtain
\[
\begin{aligned}
& \EE\left[\sup_{\xx \in [-M,M]^d} \left|G_{\Wf}(\xx)-f(\xx)\right| \right]  \leq \frac{ \EE\left[ \left| W_1 \right|^2 \right]^{\frac{1}{2}}}{\sqrt{n}}
\\ & \quad  + \frac{8\pi}{\sqrt{n}}
\left\|\frac{\widehat{f}}{\pi_a}\right\|_\infty  \left[\frac{ \EE[B_1^2]^{\frac{1}{2}}}{2} 
 + 2 M d^{\frac{1}{2}}
\EE[\|\mathbf{A}_1\|^2]^{\frac{1}{2}}
\right],
\end{aligned}
\]
where we used the boundedness of $\widehat{f}/\pi_a$ to guarantee that $w_i \leq 2 \|\frac{\widehat{f}}{\pi_a}\|_\infty$ and the functions~$\varrho_{w_i}$ can be chosen $1$-Lipschitz and the comparison theorem \cite[Theorem~4.12]{Ledoux2013} can be applied. 
We recall from the proof of Proposition~\ref{prop:FourierApproxRandom} that $\EE[|W_1|^2]^{\frac{1}{2}} =\Lbarf$ and calculate
 $\EE[B_1^2]=\frac{1}{2}$. Inserting these expressions, we obtain
 \[
 \begin{aligned}
 & \EE\left[\sup_{\xx \in [-M,M]^d} \left|G_{\Wf}(\xx)-f(\xx)\right| \right]
\\ &  \leq \frac{8\pi}{\sqrt{n}} \left\|\frac{\widehat{f}}{\pi_a}\right\|_\infty  \left(2^{-\frac{3}{2}}
 + 2 M d^{\frac{1}{2}} \sqrt{\EE[\|\mathbf{A}_1\|^2]}
 \right)  + \frac{ \Lbarf}{\sqrt{n}}.
 \end{aligned}
 \]
\end{proof}

\section*{Acknowledgment}
AJ is supported by the EPSRC grants EP/W032643/1 and  EP/T032146/1.




\bibliographystyle{IEEEtran}
\bibliography{references}

\appendix
\subsection{Construction of $\ket{\psi}$}
\label{sec:n0}

\begin{example}
	To clarify the somewhat abstract construction above, we provide some examples to make it more explicit:
	{\tiny
		\begin{center}
			\begin{tabular}{|c|c|c|c|}
				\hline
				$n$ & $n_0$ & $\nq$ & $\ket{\psi}$\\
				\hline
				1 & 0 & 2 & $\ket{00}$\\
				2 & 0 & 3 & $\displaystyle \frac{1}{\sqrt{2}}\left(\ket{000}+\ket{100}\right)$\\
				3 & 4 & 4 & $\displaystyle \frac{1}{\sqrt{3}}\left(\ket{0000}+\ket{0100}+\ket{1000}\right)$\\
				4 & 0 & 4 & $\displaystyle \frac{1}{\sqrt{4}}\left(
				\ket{0000}+\ket{0100}+\ket{1000}+ \ket{1100}\right)$\\
				5 & 12 & 5 & $\displaystyle \frac{1}{\sqrt{5}}\left(
				\ket{00000}+\ket{00100}+\ket{01000}+\ket{01100}+\ket{10000}
				\right)$\\
				\hline
			\end{tabular}
		\end{center}
	}
\end{example}
{
	In the case $n_0=0$, 
	so that $n = 2^{\nq-2}$,
	the quantum state $\ket{\psi}$
	is easy to construct explicitly:
	introduce the sequences
	$(n_l)_{l\geq 0}$ and $(\nq_l)_{l\geq 0}$ as
	$\nq_l := l+2$ and $n_l := 2^{\nq_l-2}$ (so that $4n_{l} = 2^{\nq_l}$),
	and define $\ket{\psi_l}_{\nq_l}$
	as the $\nq_l$-qubit quantum state defined as
	$$
	\ket{\psi_l}_{\nq_l}
	:= \frac{1}{\sqrt{n_l}}\sum_{i=0}^{n_l-1}\ket{4i}_{\nq_l},
	$$
	where we write 
	$\ket{4i}_{\nq_l}
	= \ket{j_{0}^{(i)}\cdots j_{\nq_l-1}^{(i)}} 
	= \ket{j_{0}^{(i)}}\otimes\cdots\otimes\ket{j_{\nq_l-1}^{(i)}}$ in binary form,
	with 
	$\displaystyle 4i = \sum_{k=0}^{\nq_l-1}j_{k}^{(i)}\cdot 2^{\nq_l-1-k}$ and $j_{k}^{(i)} \in \{0,1\}$.
	Since we are only considering multiples of~$4$, then
	$\ket{4i}_{\nq_l}
	= \ket{j_{0}^{(i)}\cdots j_{\nq_l-3}^{(i)}}\otimes\ket{00}$ for $\nq_{l}\geq 3$ and $\ket{4i}_{2} = \ket{00}$ with $i=0$.
	\begin{lemma}\label{lem:ConstrPsi}
		For any $l \in \mathbb{N}$, we have
		$$
		\ket{\psi_{l}}_{\nq_l} = \left(\bigotimes_{i=0}^{\nq_{l}-2}
		\Hg\ket{0}\right)\otimes \ket{00}.
		$$
	\end{lemma}
	\begin{proof}
		The proof follows by recursion.
		Since we are only considering multiples of~$4$ the two rightmost qubits are clearly zero and are left unchanged.
		Assume that we know 
		\begin{align*}
			\ket{\psi_{l}}_{\nq_l}
			&= \frac{1}{\sqrt{n_{l}}}\sum_{i=0}^{n_{l}-1}\ket{4i}_{\nq_l}
			= \frac{1}{\sqrt{n_{l}}}
			\sum_{i=0}^{n_{l}-1}
			\ket{j_{0}^{(i)}\cdots j_{\nq_l-1}^{(i)}}\\
			& = \frac{1}{\sqrt{n_{l}}}
			\left(\sum_{i=0}^{n_{l}-1}
			\ket{j_{0}^{(i)}\cdots j_{\nq_l-3}^{(i)}}\right)\otimes\ket{00},
		\end{align*}
		for some $l\geq 3$, with $j_{0}^{(i)} \in \{0,1\}$.
		Then
		\begin{align*}
			\ket{\psi_{l+1}}_{\nq_l+1}
			& = \frac{1}{\sqrt{n_{l+1}}}\sum_{i=0}^{n_{l+1}-1}\ket{4i}_{\nq_{l+1}}\\
			& = \frac{1}{\sqrt{n_{l+1}}}
			\left(\sum_{i=0}^{n_{l}-1}\ket{4i}_{\nq_{l+1}}
			+ \sum_{i=n_{l}}^{n_{l+1}-1}\ket{4i}_{\nq_{l+1}}\right)\\
			&  = \frac{1}{\sqrt{n_{l+1}}}
			\bigg(\sum_{i=0}^{n_{l}-1}
			\ket{0 j_{0}^{(i)}\cdots j_{\nq_l-3}^{(i)}}_{\nq_l-1}\otimes\ket{00}
			\bigg.\\
			& \qquad \qquad \bigg.+ \sum_{i=n_{l}}^{n_{l+1}-1}\ket{1j_{0}^{(i)}\cdots j_{\nq_l-3}^{(i)}}_{\nq_{l}-1}\otimes\ket{00}\bigg)\\
			& = \frac{1}{\sqrt{n_{l+1}}}
			\left(\sqrt{n_{l}}
			\ket{0}\otimes\ket{\psi_{l}}_{\nq_l}
			+ \ket{1}\otimes\sqrt{n_{l}}\ket{\psi_{l}}_{\nq_l}\right)\\
			& = (\Hg\ket{0})\otimes\ket{\psi_{l}}_{\nq_l},
		\end{align*}
		since $\frac{n_l}{n_{l+1}}=\frac{1}{2}$,
		and the lemma follows.
	\end{proof}
	While $\ket{\psi}$ is easy to construct, the unitary operator~$\Ug$ is less so, at least in an explicit form. 
	Several algorithms exist to decompose a general unitary matrix into one-qubit or two-qubit quantum gates, as in~\cite{krol2022efficient} or~\cite{li2013decomposition}
	and we leave the implementation of these algorithms to future applied work.
}
\subsection{Approximation of quantum circuits}
{
For completeness, here we prove, as a special case, that the circuit constructed in Section~\ref{sec:variational} serves as universal approximator for generic parametrized quantum circuits. Consider a quantum circuit acting on $m$ qubits via a parametric gate $\tilde{\Ug}(\tilde{\ttheta},\xx)$ for inputs $\xx \in \R^d$ and parameters $\tilde{\ttheta} \in \R^K$. Assume  that $\xx \mapsto \tilde{\Ug}(\tilde{\ttheta},\xx)$ is a measurable function $\R^d \to \C^{2^m}$ for any  $\tilde{\ttheta} \in \R^K$.  Denote by $\ket{\psi(\tilde{\ttheta},\xx)}$ the state prepared by $\tilde{\Ug}(\tilde{\ttheta},\xx) \ket{0}^{\otimes m}$, and let $\mathcal{M} = \sum_{i=1}^{2^m} \nu_i \ket{\nu_i} \bra{\nu_i}$ denote a measurement in diagonal representation.  
Then the variational quantum circuit output is 
\[
g_{\tilde{\ttheta}}(\xx) = \bra{\psi(\tilde{\ttheta},\xx)} \mathcal{M} \ket{\psi(\tilde{\ttheta},\xx)} = \sum_{i=1}^{2^m} \nu_i p_i(\tilde{\ttheta},\xx) 
\]
with $p_i(\tilde{\ttheta},\xx) = |\bra{\nu_i} \ket{\psi(\tilde{\ttheta},\xx)}|^2$ the probability of measuring outcome $\nu_i$. In particular, $p_i(\tilde{\ttheta},\xx) \in [0,1]$ and hence $g_{\tilde{\ttheta}} \in L^2(\R^d,\mu)$ for any probability measure $\mu$ on $\R^d$. Corollary~\ref{cor:universality} hence proves that for any $\varepsilon >0$ there exist $n \in \N$, $R>0$ and $\ttheta\in\TTheta$ such that 
$\Cg_{\nq}(\ttheta,\xx)$ outputs  $f_{\ttheta}(\xx)$ with 
	\begin{equation*}
		\left(\int_{\RR^d} |g_{\tilde{\ttheta}}(\xx)-f_{\ttheta}(\xx)|^2 \mu(\D\xx)\right)^{1/2} \leq  \varepsilon.
	\end{equation*}
Thus, the circuits constructed in Section~\ref{sec:variational} are capable of approximating arbitrarily well the outputs of generic parametrized quantum circuits.
Note that this statement concerns the variational quantum circuit as input-output mapping, rather than a direct approximator of generic quantum gates. 
}
{
\subsection{Monte Carlo error}
\label{appendixC}
In our definition of the quantum neural networks in~\eqref{eq:qnn} and~\eqref{eq:randomF} we work with probabilities, rather than their Monte Carlo estimates~\eqref{eq:probestimate}. 
This simplification is not a limitation though, as the Monte Carlo error can be incorporated directly. 
Indeed, denote by $f_{n,\ttheta}^{R,S}$ the quantum neural network with output probabilities estimated by $S$ shots. Recall that  for i.i.d.\ random variables $X_1,\ldots,X_S$ then $\EE[| \EE[X_1]
- \frac{1}{S} \sum_{s=1}^S X_s |^2] = \mathrm{Var}(\frac{1}{S} \sum_{s=1}^S X_s ) = \frac{\mathrm{Var}(X_1)}{S}$, and hence
		\begin{equation*}
	\begin{aligned}
		&\EE\left[\int_{\RR^d} 
		\left|f_{n,\ttheta}^{R}(\xx) - f_{n,\ttheta}^{R,S}(\xx)\right|^2 \mu(\D \xx)\right]^{1/2}
		\\& 
		\leq  2R \sum_{j=1}^2 \left(\int_{\RR^d} 
		\EE\left[\left| \P_j^{n}(\ttheta,\xx)
		- \widehat{\P}_j^{n}(\ttheta,\xx)\right|^2\right] \mu(\D \xx)\right)^{1/2}
		\\& 
		\leq  \frac{4R}{\sqrt{S}}.
	\end{aligned}
\end{equation*} 
Minkowski's inequality and Theorem~\ref{thm:Approx} therefore yield an approximation error bound for the network $f_{n,\ttheta}^{R,S}$: 
		\begin{equation*}
			\begin{aligned}
				&\left(\EE\left[\int_{\RR^d} 
				\left|f(\xx) - f_{n,\ttheta}^{R,S}(\xx)\right|^2 \mu(\D \xx)\right]\right)^{1/2}
				\leq  \frac{\LOnefhat}{\sqrt{n}} + \frac{4R}{\sqrt{S}}.
			\end{aligned}
		\end{equation*}
The same argument also applies to the network \eqref{eq:randomF}. 
}


%






\end{document}